\documentclass[11pt,reqno]{amsart}

\usepackage{amsmath,amsthm,amssymb,amsfonts}
\usepackage{bm}
\usepackage{natbib}
\usepackage{bbm}
\usepackage{graphicx}
\usepackage{multirow}
\usepackage{longtable}
\usepackage{here}
\usepackage{fullpage}
\usepackage{xurl}
\usepackage{color}
\usepackage{xcolor}
\usepackage{mathrsfs}
\usepackage{enumitem}
\usepackage{comment}
\usepackage{subcaption}
\usepackage{pifont}
\usepackage[linesnumbered,ruled,vlined]{algorithm2e}

\makeatletter

\@addtoreset{equation}{section}
\makeatletter

\newtheorem{theorem}{Theorem}[section]

\newtheorem{proposition}{Proposition}[section]
\newtheorem{assumption}{Assumption}[section]
\theoremstyle{definition}
\newtheorem{definition}{Definition}[section]
\newtheorem{remark}{Remark}[section]
\newtheorem{example}{Example}[section]

\DeclareMathOperator{\Var}{Var}

\DeclareMathOperator{\E}{E}
\DeclareMathOperator{\Prob}{P}

\def\f{Fr\'echet }
\DeclareMathOperator*{\argmin}{arg\,min}
\newcommand{\Eo}{\E_\oplus}
\renewcommand{\tilde}{\widetilde}
\renewcommand{\hat}{\widehat}

\renewcommand{\tilde}{\widetilde}
\renewcommand{\hat}{\widehat}

\def\f{Fr\'echet }
\def\bco{\iffalse} 
\allowdisplaybreaks

\allowdisplaybreaks

\makeatother

\begin{document}

\title[]{Geodesic Causal Inference}
\thanks{The work of D. K. was partially supported by JSPS KAKENHI Grant Numbers 23K12456 and 26K00323 and the work of H.G.M. was partially supported by NSF grant DMS-2310450} 

\author[D. Kurisu]{Daisuke Kurisu$^*$}
\author[Y. Zhou]{Yidong Zhou$^*$}
\author[T. Otsu]{Taisuke Otsu}
\author[H.-G. M\"uller]{Hans-Georg M\"uller}

\renewcommand{\thefootnote}{\fnsymbol{footnote}}
\footnotetext[1]{The first two authors contributed equally to this work and are listed alphabetically.}

\address[D. Kurisu]{Faculty of Economics, The University of Tokyo\\
7-3-1, Hongo, Bunkyo-ku, Tokyo 113-0033, Japan.
}
\email{d.kurisu@e.u-tokyo.ac.jp}

\address[Y. Zhou]{Department of Statistics, University of California, Davis, One Shields Avenue, Davis, CA, 95616, USA.
}
\email{ydzhou@ucdavis.edu}

\address[T. Otsu]{Department of Economics, London School of Economics, Houghton Street,
London, WC2A 2AE, UK.
}
\email{t.otsu@lse.ac.uk}

\address[H.-G. M\"uller]{Department of Statistics, University of California, Davis, One Shields Avenue, Davis, CA, 95616, USA.
}
\email{hgmueller@ucdavis.edu}

\begin{abstract}
Adjusting for confounding and imbalance when establishing statistical relationships is an increasingly important task, and causal inference methods have emerged as the most popular tool to achieve this. Existing methodology has been developed primarily for outcomes that lie in Euclidean spaces. We introduce here a general framework for causal inference when outcomes reside in general geodesic metric spaces, where we draw on a novel geodesic calculus that facilitates scalar multiplication for geodesics and the quantification of treatment effects through the concept of geodesic average treatment effect. Using ideas from Fr\'echet regression, we obtain a doubly robust estimation of the geodesic average treatment effect and results on consistency and rates of convergence for the proposed estimators. We also develop an intrinsic uncertainty quantification framework for the treatment effect based on Fr\'echet objective functions. The proposed framework is illustrated through simulations and real data applications, including network-valued outcomes from New York City taxi trips to assess the impact of the COVID-19 pandemic, and compositional data on U.S. state-level energy sources to study the effect of coal mining. 
\end{abstract}

\maketitle

\section{Introduction}\label{Sec:Intro}
Causal inference aims to mitigate the effects of confounders in the assessment of treatment effects and is an intensively studied research topic \citep{pete:17, ding:23}. While inadequate adjustment for confounders leads to biased causal effect estimates, \citet{rose:83} demonstrated in a seminal paper that the propensity score, defined as the probability of treatment assignment conditional on confounders, can be used to correct this bias. \citet{robin:94} proposed a \emph{doubly robust method} combining outcome regression and propensity score modeling. Doubly robust estimators enjoy favorable theoretical properties under the correct specification of either the outcome regression or propensity score models \citep{scha:99, bang:05, ogbu:15} and have become a standard for treatment effect estimation and inference.

Almost all existing work in causal inference has focused on scalar outcomes, or more generally, outcomes situated in a Euclidean space. A notable exception is the framework of \citet{LiKoWa23}, which provides causal inference for univariate distributional outcomes in the Wasserstein space by exploiting that quantile functions form a convex subset of $L^2$. Related work on functional outcomes \citep{ecke:24} considers causal inference in Hilbert spaces using outcome regression, but its reliance on linear structure limits applicability to settings where the outcome resides in a non-Euclidean space.

As pointed out by a referee, an earlier arXiv version of \citet{lin:21} included a broader extension to manifold-valued outcomes using an extrinsic approach based on tangent-space linearization. However, such extrinsic approaches face two fundamental limitations. First, many outcome spaces encountered in practice do not admit a smooth manifold structure with well-defined tangent spaces and invertible exponential and logarithmic maps, including spheres with boundaries, tree spaces \citep{bill:01}, and other general Hadamard spaces. Second, even when such maps exist, extrinsic methods suffer from a non-revertibility problem: linear operations performed in the tangent space need not correspond to valid objects after mapping back. This issue arises even for univariate distributions in the Wasserstein space, where the logarithmic map takes values in a convex cone rather than a linear subspace and the exponential map cannot undo arbitrary linear combinations without projection, leading to geometric distortions such as tail shrinkage and support truncation \citep{bigo:17,caze:18,pego:22,chen:23}. These limitations restrict the applicability and interpretability of extrinsic approaches.

Here we aim at a major generalization to obtain causal inference and specifically doubly robust treatment effect estimation for outcomes that are located in a geodesic metric space using an intrinsic approach. Euclidean, functional, and univariate distributional outcomes are included as special cases, and our framework accommodates a wide range of random objects encountered in applications, including networks, symmetric positive-definite (SPD) matrices and compositional data. Examples in Section~\ref{Sec:Pre1} and the real-data analyses demonstrate the practical relevance of this extension. A central challenge in this setting is that general metric spaces lack linear operations such as addition, scalar multiplication, and inner products, which underpin classical causal inference methodology. As a result, algebraic notions of differences or averages cannot be used to quantify treatment effects and must be replaced with intrinsically defined geometric constructs. To address this, we work in uniquely geodesic metric spaces that are extendable and draw on novel developments in geodesic calculus to define a notion of scalar multiplication along geodesics. This structure provides a principled analogue of linear interpolation and forms the basis for defining treatment effects on complex outcomes.

\begin{figure}[tb]
    \centering
    \includegraphics[width=0.5\linewidth]{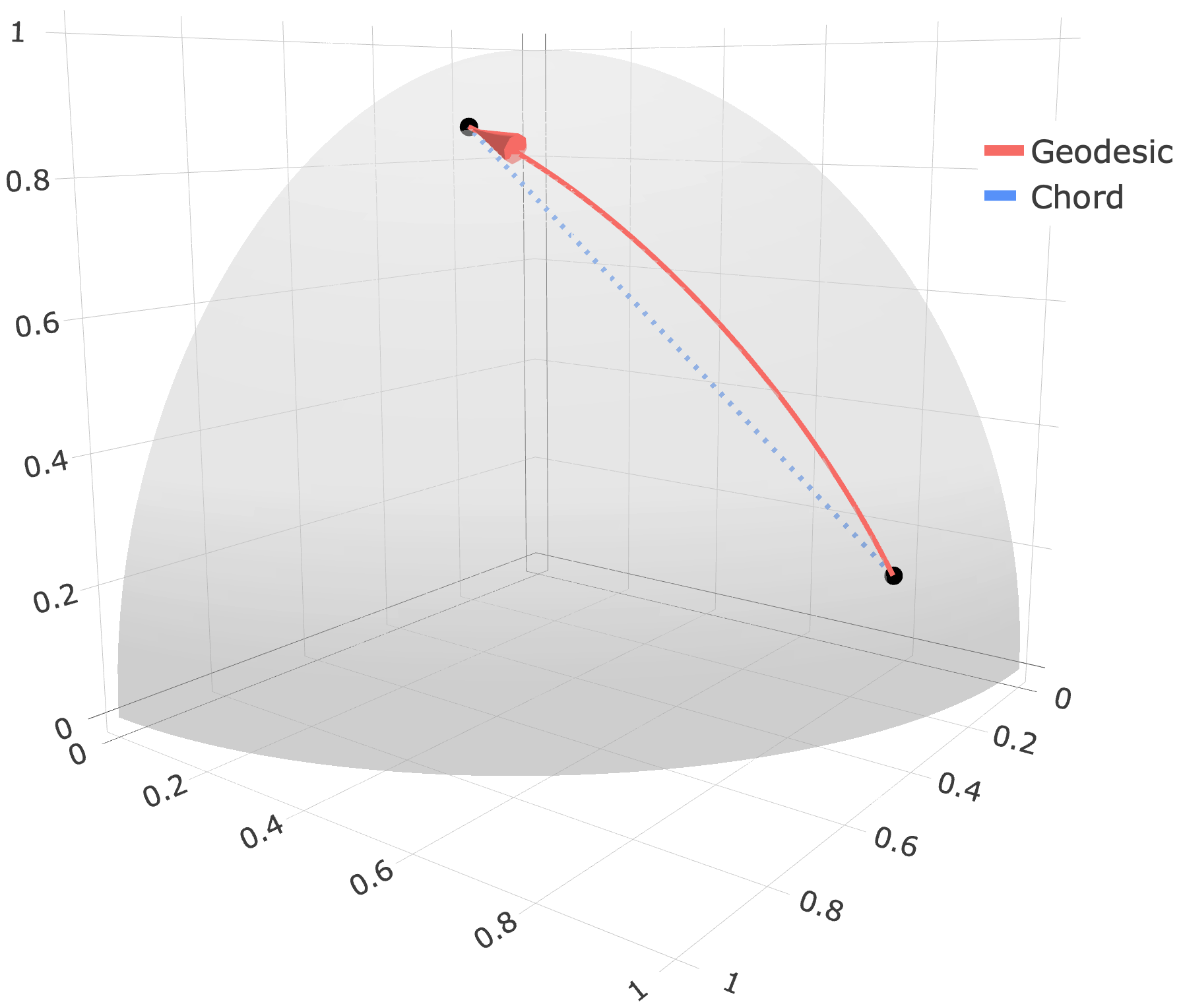}
    \caption{Geodesic (red solid) and Euclidean chord (blue dashed) between two points on $\mathcal S^2_{+}$, the positive orthant of the unit sphere. Many complex outcomes, including compositions after the square-root transformation, lie in this space. The geodesic follows the intrinsic shortest path, while the chord ignores the underlying geometry.}
    \label{fig:s2illus}
\end{figure}

A key contribution of our work is the \emph{geodesic average treatment effect} (GATE), defined as the \emph{entire geodesic curve} connecting the mean potential outcomes. In Euclidean settings, the GATE reduces to the line segment joining the two mean potential outcomes. In nonlinear spaces, however, the full geodesic captures both the \emph{magnitude} and the \emph{direction} of the change induced by treatment, providing insight into how the underlying object evolves along the shortest intrinsic path. This goes beyond scalar summaries such as the absolute average treatment effect of \citet{shin:24}, which quantify only the size of the effect. Their work, which appeared shortly after the initial circulation of our manuscript, focuses on distance-based estimands and matching strategies. By contrast, the GATE preserves the full geometric trajectory between potential outcomes and is amenable to doubly robust estimation, offering a more informative and flexible representation of treatment effects for complex outcomes.

The interpretability of the GATE arises from the fact that it reflects the intrinsic geometry of the outcome space. Figure~\ref{fig:s2illus} illustrates this in $\mathcal S^2_{+}$, the positive orthant of the unit sphere. Certain data types naturally reside in this space; for example, Example~\ref{exm:com} shows that compositional data equipped with the arclength metric can be viewed as elements of $\mathcal S^2_{+}$ \citep{scea:11,scea:14}. The red curve depicts the spherical geodesic, the shortest path compatible with the geometry, while the blue dotted segment is the Euclidean chord that ignores this structure. Much like a great-circle route on the globe provides the true shortest path between two locations, the geodesic captures the intrinsic direction of change between the two objects. This viewpoint extends naturally to other complex outcomes. For distributional data under the Wasserstein metric, the geodesic corresponds to mass transport; for networks, it describes systematic deformation of connectivity patterns; and for SPD matrices, it traces smooth transitions in covariance structure. Across these settings, the geodesic provides a coherent and interpretable description of how complex outcomes change.

A tempting but inadequate strategy for such data is to ignore the underlying geometry and analyze each coordinate or matrix entry separately. This treats the object as a collection of independent scalars and often violates its structural constraints. For compositional data, elementwise modeling can break the unit-sum constraint and yield invalid compositions. For networks, entrywise changes need not preserve connectivity structure, and summaries such as degrees or clustering coefficients capture only fragments of the overall object. For SPD matrices, coordinatewise adjustments can destroy positive-definiteness. In contrast, the GATE defines the treatment effect through the intrinsic geodesic between the mean potential outcomes, ensuring that every point along the path remains a valid object of the same type and capturing the coherent, geometry-respecting deformation induced by treatment.

Building on this intrinsic representation, we construct doubly robust estimators for the GATE and establish consistency and convergence rates. We further develop an intrinsic uncertainty quantification framework based on Fr\'echet objective functions, enabling the construction of test statistics, confidence regions, and hypothesis tests directly in the metric space. The finite-sample performance of the proposed estimators is examined through simulations, and applications to compositional data and network-valued outcomes highlight their practical utility.

The remainder of the paper is organized as follows. Section~\ref{Sec:Pre} reviews geodesic metric spaces, introduces scalar multiplication on geodesics, presents examples of relevant outcome spaces, and defines the GATE. Section~\ref{sec:GATE} develops doubly robust estimators for the GATE. Section~\ref{sec:DR} establishes their asymptotic properties, including consistency and convergence rates. Section~\ref{sec:uq} develops intrinsic uncertainty quantification procedures based on Fr\'echet objective functions. Section~\ref{sec:sim} reports simulation results for SPD matrices and compositional data, and Section~\ref{sec:data} illustrates the methodology through applications to taxi network data and electricity generation compositions. Section~\ref{sec:disc} concludes with a discussion of potential extensions and open questions.

\section{Treatment effects for random objects}\label{Sec:Pre}
\subsection{Preliminaries on geodesic metric spaces} \label{Sec:Pre1}
Consider a metric space \((\Omega, d)\) equipped with a probability measure $\Prob$. A \textit{curve} in \(\Omega\) is a continuous map \(\gamma : [0, 1] \to \Omega\) with length $L(\gamma) = \sup\sum_{i=0}^{I-1} d\{\gamma(t_{i}),\gamma(t_{i+1})\}$, where the supremum is taken over all possible partitions of the interval \([0, 1]\) with breakpoints \(0 = t_0 \leq t_1 \leq \cdots \leq t_I = 1\). A curve \(\gamma\) is called a \textit{geodesic} if it satisfies \(d(\gamma(s), \gamma(t)) = |t - s|d(\gamma(0), \gamma(1))\) for every \(s, t \in [0, 1]\). The space \((\Omega, d)\) is termed a \textit{geodesic metric space} if every pair of points \(\alpha, \beta \in \Omega\) is connected by a geodesic, denoted as \(\gamma_{\alpha, \beta}\) \citep{bura:01}.

We assume that \((\Omega, d)\) is a \textit{uniquely geodesic metric space}, meaning that for every pair of points in \(\Omega\), there exists exactly one geodesic connecting them. For \(\alpha, \beta, \zeta \in \Omega\), the following operations are defined for geodesics:
\begin{align*}
\gamma_{\alpha,\zeta}\oplus\gamma_{\zeta,\beta} & :=\gamma_{\alpha,\beta},\ \ominus \gamma_{\alpha,\beta} := \gamma_{\beta,\alpha},\ \mathrm{id}_{\alpha} :=\gamma_{\alpha,\alpha}.
\end{align*}

Additionally, we assume that \((\Omega, d)\) is \textit{extendable}, meaning that every geodesic can be extended to a geodesic line defined on \(\mathbb{R}\) \citep{brids:99}. Specifically, for any \(\alpha, \beta \in \Omega\), the geodesic \(\gamma_{\alpha, \beta} : [0, 1] \to \Omega\) can be extended to \(\bar{\gamma}_{\alpha, \beta} : \mathbb{R} \to \Omega\), such that the restriction of \(\bar{\gamma}_{\alpha, \beta}\) to \([0, 1]\) coincides with \(\gamma_{\alpha, \beta}\), i.e., \(\bar{\gamma}_{\alpha, \beta}|_{[0, 1]} = \gamma_{\alpha, \beta}\). We can then define a scalar multiplication for geodesics $\gamma_{\alpha, \beta}$ with a factor $\rho\in\mathbb{R}$ by
\begin{align*}
\rho \odot \gamma_{\alpha,\beta} & :=
\begin{cases}
\{\bar{\gamma}_{\alpha,\beta}(t):t\in[0,\rho]\} & \text{if $\rho>0$}\\
\{\bar{\gamma}_{\alpha,\beta}(t):t\in[\rho,0]\} & \text{if $\rho<0$}\\
\mathrm{id}_\alpha & \text{if $\rho=0$} 
\end{cases},\quad\rho\odot\mathrm{id}_{\alpha}:=\mathrm{id}_{\alpha}.
\end{align*}

In practical applications, the relevant metric space often constitutes a subset $(\mathcal{M}, d)$ of $(\Omega, d)$, which is typically closed and convex. Prominent instances of such spaces are given in Examples \ref{exm:net}--\ref{exm:com} below. We use these same example spaces to demonstrate the practical performance of the proposed approach in simulations and real-world data analyses. In Appendix B, we demonstrate that these examples satisfy the pertinent assumptions outlined in Sections \ref{sec:GATE} and \ref{sec:DR}. 

Geodesic extensions are confined to operate within the closed subset $\mathcal{M} \subset \Omega$, in some cases necessitating suitable modifications to ensure that they do not cross the boundary of $\mathcal{M}$ \citep{zhu:23}. Considering the forward extension of a geodesic where $\rho>1$, assume that the geodesic $\gamma_{\alpha,\beta}$ extends to the boundary point $\zeta$ and denote the extended geodesic by $\gamma_{\alpha, \zeta}: [0, 1]\mapsto\mathcal{M}$. We then define a scalar multiplication for $\rho>1$ by 
\[\rho\odot\gamma_{\alpha, \beta}=\{\gamma_{\alpha, \zeta}(t): t\in[0, h(\rho)]\},\]
where 
\[h(\rho)=-(1-\frac{d(\alpha, \beta)}{d(\alpha, \zeta)})^\rho+1.\]
It is straightforward to verify that the start point of $\rho\odot\gamma_{\alpha, \beta}$ is always $\alpha$, while the end point is $\beta$ for $\rho=1$ and $\zeta$ for $\rho=\infty$; see Figure~\ref{fig:ext} for an illustration. Similar modifications apply to the reverse extension where $\rho<-1$. To simplify the notation, we write the end point of $\rho\odot\gamma_{\alpha, \beta}$ as $\gamma_{\alpha, \beta}(\rho)$.

\begin{figure}[tb]
    \centering
    \includegraphics[width=0.45\linewidth]{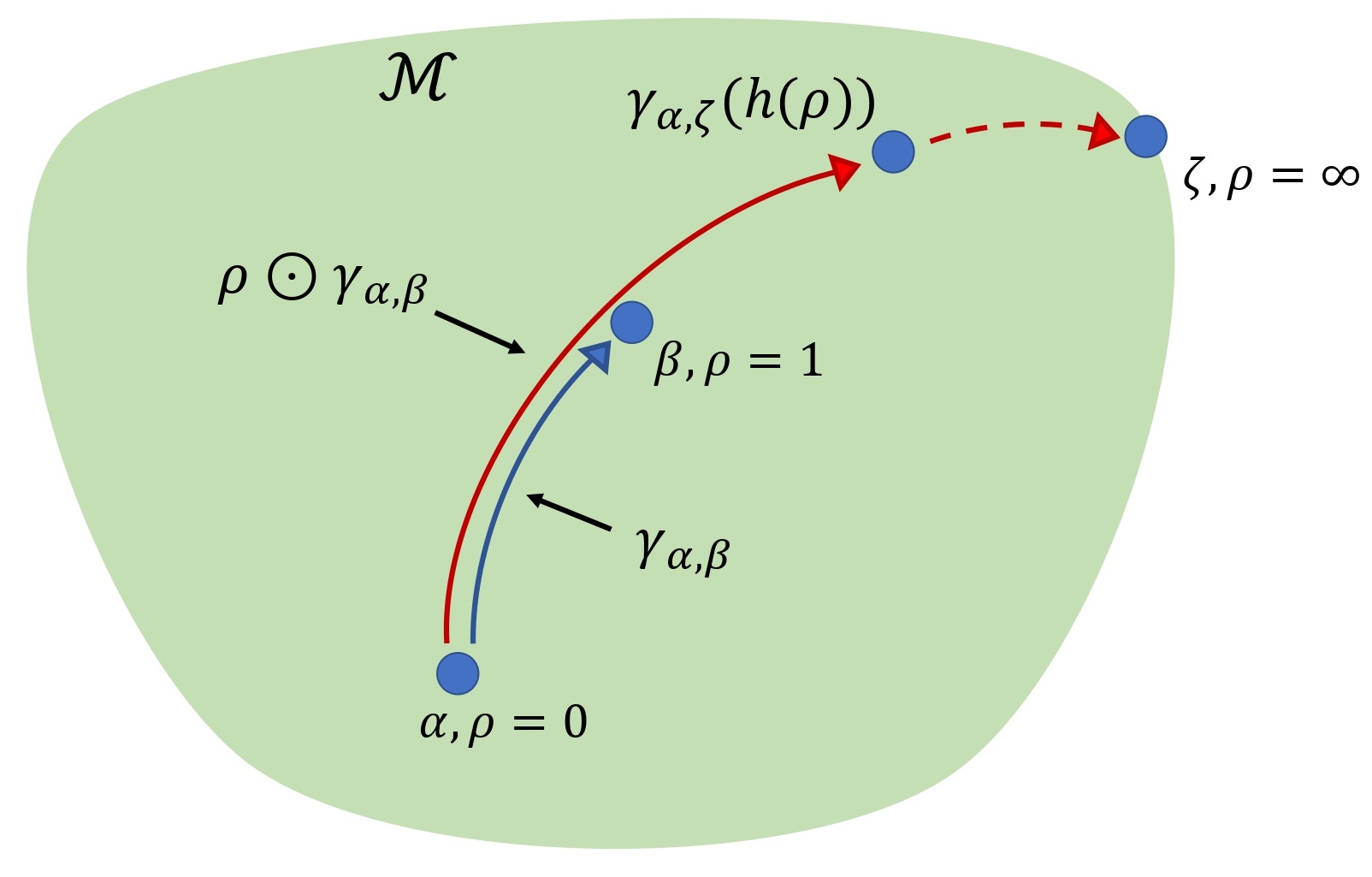}
    \caption{Illustration of geodesic extension. The circles symbolize objects in the geodesic metric space $\mathcal{M}$. The arrows emanating from these circles represent the geodesic paths connecting them.}
    \label{fig:ext}
\end{figure}

\begin{example}[Networks]
\label{exm:net}
Network data represent interactions among units such as individuals in a social system, locations in a transportation network, or genes in a regulatory pathway. A common analytic representation is the graph Laplacian \citep{kola:14,zhou:22,seve:22}. For an undirected, weighted graph with $m$ nodes and bounded edge weights, its Laplacian lies in a convex subset of $\mathbb{R}^{m^2}$. When equipped with the Frobenius metric $d_F(L_1, L_2) = \|L_1 - L_2\|_F$, geodesics between Laplacians are given by the linear interpolation $\gamma_{L_1,L_2}(t) = (1 - t)L_1 + tL_2, t \in [0,1]$.
\end{example}

\begin{example}[Symmetric positive-definite matrices]
\label{exm:cor}
Symmetric positive-definite (SPD) matrices are frequently used to represent structured second-order information, such as covariance patterns, diffusion tensors, or kernel matrices. Let $\mathrm{Sym}_m^{+}$ denote the space of $m\times m$ SPD matrices. Several geodesic metrics are commonly used on this space. Under the Frobenius metric $d_F(A,B)=\|A-B\|_F$, the geodesic between $A$ and $B$ is the linear path $\gamma_{A, B}(t)=(1-t)A+tB$. The Log-Euclidean metric \citep{arsi:07}, defined by $d_{LE}(A,B)=\|\log A-\log B\|_F$, instead yields the geodesic $\gamma_{A, B}(t)=\exp\{(1-t)\log A+t\log B\}$. Both metrics endow $\mathrm{Sym}_m^{+}$ with a uniquely geodesic structure, and related constructions such as the power metric \citep{dryd:09} and the Log-Cholesky metric \citep{lin:19:1} provide additional geometries with similar properties.
\end{example}

\begin{example}[Functional data]
\label{exm:fun}
Functional data refer to observations that take the form of smooth curves or trajectories, often recorded over time or across a spatial domain \citep{rams:05,hsin:15,mull:16:3}. Let $\mathcal{T}$ be a compact interval and consider the space $L^{2}(\mathcal{T})$ of square-integrable functions on this domain. This space carries the inner product $\langle f,g\rangle=\int_{\mathcal{T}} f(t)g(t)\,dt$ and the corresponding metric $d_{L^{2}}(f,g)=\left(\int_{\mathcal{T}}\{f(t)-g(t)\}^{2}\,dt\right)^{1/2}$. Equipped with this structure, $L^{2}(\mathcal{T})$ is a Hilbert space, and geodesics are simply linear blends of functions, given by $\gamma_{f,g}(t)=(1-t)f+tg$ for $t\in[0,1]$.
\end{example}

\begin{example}[Univariate probability measures]
\label{exm:mea}
Consider univariate probability measures on a closed interval $\mathcal{I}\subset\mathbb{R}$ with finite second moments. The space of such probability measures, equipped with the Wasserstein metric $d_{\mathcal{W}}$, is known as the Wasserstein space $(\mathcal{W}, d_{\mathcal{W}})$, which is a complete and separable metric space \citep{ambr:08}. The Wasserstein metric between any two probability measures $\mu$ and $\nu$ is
\[d_{\mathcal{W}}^2(\mu, \nu)=\int_0^1\{F_{\mu}^{-1}(p)-F_{\nu}^{-1}(p)\}^2dp,\]
where $F_{\mu}^{-1}(\cdot)$, $F_{\nu}^{-1}(\cdot)$ denote the quantile functions of $\mu$ and $\nu$, respectively. Write $\tau_\#\mu$ for the pushforward measure of $\mu$ by the transport $\tau$. The geodesic between $\mu$ and $\nu$ is given by McCann's interpolant \citep{mcca:97}, $\gamma_{\mu, \nu}(t)=(\mathrm{id}+t(F_\nu^{-1}\circ F_\mu-\mathrm{id}))_\#\mu$, where $\mathrm{id}$ and $F_\mu$ denote the identity map and cumulative distribution function of $\mu$, respectively.
\end{example}

\begin{example}[Compositional data]
\label{exm:com}
Compositional data takes values in the simplex
\[\Delta^{d-1}=\{\bm{y}\in\mathbb{R}^d: y_j\geq0, j=1, \ldots, d,\text{ and }\sum_{j=1}^dy_j=1\},\]
reflecting that such data are non-negative proportions that sum to 1. Consider the component-wise square root $\sqrt{\bm{y}}=(\sqrt{y_1}, \ldots, \sqrt{y_d})'$ of $\bm{y}\in\Delta^{d-1}$, the simplex $\Delta^{d-1}$ can be mapped to the first orthant of the unit sphere $\mathcal{S}^{d-1}_+=\{\bm{z}\in\mathcal{S}^{d-1}:z_j\geq0, j=1, \ldots, d\}$ \citep{scea:11,scea:14}. Equipping $\mathcal{S}^{d-1}_+$ with the geodesic (Riemannian) metric on the sphere $d_g(\bm{z}_1, \bm{z}_2)=\mathrm{arccos}(\bm{z}_1'\bm{z}_2)$ for $\bm{z}_1, \bm{z}_2\in\mathcal{S}^{d-1}_+$ induces a uniquely geodesic structure. The geodesic connecting $\bm{z}_1$ and $\bm{z}_2$ is 
\[\gamma_{\bm{z}_1, \bm{z}_2}(t)=\cos(\theta t)\bm{z}_1+\sin(\theta t)\frac{\bm{z}_2-(\bm{z}_1'\bm{z}_2)\bm{z}_1}{\|\bm{z}_2-(\bm{z}_1'\bm{z}_2)\bm{z}_1\|},\]
where $\theta=\arccos(\bm{z}_1'\bm{z}_2)$ represents the angle between $\bm{z}_1$ and $\bm{z}_2$.

This idea can be extended to distributions, where one considers a separable Hilbert space $\mathcal{H}$ with inner product $\langle\cdot, \cdot\rangle_{\mathcal{H}}$ and norm $\|\cdot\|_{\mathcal{H}}$. The Hilbert sphere $\mathcal{S}=\{g\in\mathcal{H}: \|g\|_{\mathcal{H}}=1\}$ is an infinite-dimensional extension of finite-dimensional spheres. The space of (multivariate) absolutely continuous distributions can be equipped with the Fisher-Rao metric $d_R(f_1, f_2)=\mathrm{arccos}(\langle\sqrt{f_1}, \sqrt{f_2}\rangle_{\mathcal{H}})$ and then modeled as a subset of the Hilbert sphere $\mathcal{S}$, where the Fisher-Rao metric pertains to the geodesic metric between square roots of densities $f_1$ and $f_2$.
\end{example}

\subsection{Geodesic average treatment effects}\label{Sec:Pre2}
Let $(\mathcal{M},d)$ be a uniquely extendable geodesic metric space that is complete and totally bounded. Units $i=1,\dots,n$ are either assigned to treatment or control. For the $i$-th unit, we observe a treatment indicator $T_i \in \{0,1\}$, where $T_i=1$ indicates treatment and $T_i=0$ otherwise. Let $Y_i(0), Y_i(1) \in \mathcal{M}$ denote the potential outcomes under control and treatment, respectively. We assume the stable unit treatment value assumption (SUTVA), which entails no interference between units and no hidden versions of treatment, and implies the consistency relation
\begin{align*}
Y_i =
\begin{cases}
Y_i(0) & \text{if } T_i = 0,\\
Y_i(1) & \text{if } T_i = 1.
\end{cases}
\end{align*}
In addition, we observe Euclidean covariates $X_i \in \mathbb{R}^p$. We assume that $\{(Y_i, T_i, X_i)\}_{i=1}^n$ is an i.i.d. sample from a joint distribution $\Prob$ on $\mathcal{M} \times \{0,1\} \times \mathcal{X}$, where a generic random variable is $(Y,T,X)$, all conditional distributions are well-defined, expectations and conditional expectations are taken with respect to $\Prob$, and $\mathcal{X}$ is assumed to be a compact subset of $\mathbb{R}^p$.

For scalar outcomes, the average treatment effect is defined as $\E[Y(1)]-\E[Y(0)]$, which is expressed as the difference between expected potential outcomes. This formulation relies fundamentally on two operations: expectation and subtraction. However, in general metric spaces, these operations are not directly available due to the absence of linear structure and must be replaced by their geometric counterparts.

Expectation for metric space-valued outcomes is generalized by the \f mean \citep{frec:48},
\[\E_\oplus[Y] = \argmin_{\nu \in \mathcal{M}} \E[d^2(\nu, Y)],\]
provided the minimizer exists and is unique. Such properties are guaranteed under suitable geometric conditions, for example in Hadamard spaces \citep{stur:03}. In the presence of covariates $X \in \mathbb{R}^p$, the corresponding conditional \f mean \citep{PeMu19} is given by
\[\E_\oplus[Y|X] = \argmin_{\nu \in \mathcal{M}} \E[d^2(\nu, Y)|X],\]
where the expectation is taken with respect to the conditional distribution of $Y$ given $X$. This definition recovers the classical conditional expectation when $\mathcal{M} = \mathbb{R}$ and $d$ is the Euclidean distance, and serves as the natural regression target for metric space-valued responses.

The second ingredient in defining causal effects is the notion of difference. In Euclidean space, the difference $\E[Y(1)] - \E[Y(0)]$ can be interpreted geometrically as the straight-line segment connecting $\E[Y(0)]$ to $\E[Y(1)]$. In a uniquely geodesic metric space, the analogue of a straight line is a geodesic, i.e., the shortest path connecting two points. Therefore, the comparison between mean potential outcomes is naturally represented by the geodesic connecting their \f means. Figure~\ref{fig:s2illus} provides an illustration of geodesics on the positive orthant of the unit sphere $\mathcal{S}^2_{+}$. These considerations lead to the following definition of causal effects in geodesic metric spaces.
\begin{definition}[Geodesic individual and average treatment effects]\label{def:GATE}
The geodesic individual treatment effect of $T$ on $Y$ for unit $i$ is defined as the geodesic connecting the two potential outcomes, $\gamma_{Y_i(0), Y_i(1)}$. The \emph{geodesic average treatment effect} (GATE) is defined as the geodesic connecting the \f means of the potential outcomes, $\gamma_{\E_\oplus[Y(0)], \E_\oplus[Y(1)]}$.
\end{definition}
When $(\mathcal{M}, d) = (\mathbb{R}, |\cdot|)$, geodesics reduce to line segments, and the definitions above recover the classical treatment effects: the individual treatment effect $Y_i(1) - Y_i(0)$ and the average treatment effect $\E[Y(1)] - \E[Y(0)]$. More generally, a geodesic connecting two points in $(\mathcal{M}, d)$ encodes both the minimal distance and directional information between them with respect to the metric. Consequently, Definition~\ref{def:GATE} provides a natural geometric generalization of treatment effects, capturing both magnitude and direction of causal impact in spaces where linear structure is absent.

\begin{remark}
    As pointed out by a referee, an alternative formulation of the GATE can be defined via conditional \f means, namely as $\gamma_{\E_{\oplus}[m_0(X)],\E_{\oplus}[m_1(X)]}$, where $m_t(x)=\E_{\oplus}[Y(t)|X=x]$. Under the non-Euclidean analogue of the law of iterated expectations (cf. Assumption~\ref{Ass:p-map}(i)), this coincides with Definition~\ref{def:GATE}. This perspective also suggests an estimation strategy based on doubly robust pseudo-outcomes applied to the scalar conditional risk functions $M_t(\nu,x)=\E[d^2(\nu,Y(t))|X=x]$ for each fixed $\nu$, followed by a minimization step to recover $m_t(x)$ and an outer empirical \f mean; see \citet{kenn:23, kenn:24}. While this yields doubly robust estimation of $M_t(\nu,x)$ for each $\nu$, the resulting estimator of $m_t(x)=\argmin_{\nu\in\mathcal{M}} M_t(\nu,x)$ requires additional control of the estimated risk surface over $\nu$ and stability of the minimizer. Moreover, it entails solving optimization problems over $\mathcal{M}$ for each $x$, which can be nontrivial in general geodesic metric spaces where no canonical parametrization or discretization is available. In contrast, the proposed formulation directly targets $\E_{\oplus}[Y(t)]$ and yields a unified geometric doubly robust estimator without introducing an intermediate optimization over $\nu$.
\end{remark}

\section{Doubly robust estimation of the GATE}\label{sec:GATE}

\subsection{Doubly robust identification}\label{subsec:GATE1}
We now develop a doubly robust representation for the GATE. The key idea is to reinterpret the classical Euclidean doubly robust estimand in geometric terms and then extend this structure to general geodesic metric spaces by replacing linear operations with their intrinsic geometric counterparts. We begin with the standard assumptions that enable identification of causal effects from observational data.

\begin{assumption}\label{Ass:GATE} \quad
\begin{itemize}
\item[(i)] There exists a constant $\eta_0 \in (0,1/2)$ such that $\eta_0 \leq p(x) \leq 1-\eta_0$ for each $x \in \mathcal{X}$. 
\item[(ii)] $T$ and $\{Y(0), Y(1)\}$ are conditionally independent given $X$. 
\end{itemize}
\end{assumption}

Assumptions \ref{Ass:GATE}(i) and (ii) correspond to the standard overlap and unconfoundedness conditions in causal inference. The overlap condition ensures that, for any covariate value $x$, both treatment and control groups are sufficiently represented, while unconfoundedness guarantees that, conditional on $X$, the treatment assignment is independent of the potential outcomes.

Let $p(x)=\Prob(T=1|X=x)$ denote the propensity score. For $t \in \{0,1\}$, define the conditional \f means of the potential outcomes as $m_t(x)=\E_{\oplus}[Y(t)|X=x]$, which play the role of outcome regression functions in the present setting. To motivate the construction in metric spaces, it is useful to recall the corresponding formulation in the Euclidean case. Let $e(x)$ and $\mu_t(x)$, $t \in \{0,1\}$, be models for the propensity score $p(x)$ and the outcome regression functions $m_t(x)$, respectively. Under Assumption~\ref{Ass:GATE}, the mean potential outcomes satisfy the classical doubly robust identity
\begin{equation}\label{eq:DR-Euclid}
\E[Y(t)] = \E\left[\mu_t(X) + \left(\frac{tT}{e(X)} + \frac{(1-t)(1-T)}{1-e(X)}\right)\big(Y - \mu_t(X)\big)\right].
\end{equation}
This identity is doubly robust in the sense that it remains valid if either the propensity score model or the outcome regression model is correctly specified.

A key observation is that the Euclidean pseudo-outcome admits a geometric interpretation. In Euclidean space, geodesics coincide with line segments, and for any $a,y \in \mathbb{R}$ and $\rho \in \mathbb{R}$, one can write $a + \rho (y-a) = \gamma_{a,y}(\rho)$. Therefore, \eqref{eq:DR-Euclid} can be equivalently written as
\[\E[Y(t)] = \E\left[\gamma_{\mu_t(X),Y}\left(\frac{tT}{e(X)} + \frac{(1-t)(1-T)}{1-e(X)}\right)\right].\]
This formulation reveals that the doubly robust estimator can be viewed as an average of points obtained by appropriately weighting (and potentially extrapolating along) the geodesic from $\mu_t(X)$ toward $Y$. This geometric perspective suggests a natural extension to general metric spaces: replace Euclidean line segments by geodesics and replace expectations by \f means.

To carry out this extension, we require structural properties of the \f mean that parallel those of Euclidean expectation.

\begin{assumption}\label{Ass:p-map}\quad
\begin{itemize}
    \item[(i)] For any random object $Y \in \mathcal{M}$ and random vector $X \in \mathbb{R}^p$, the \f mean satisfies $\E_\oplus[Y] = \E_\oplus[\E_\oplus[Y|X]]$.
    \item[(ii)] For any fixed $\alpha \in \mathcal{M}$, random object $Y \in \mathcal{M}$, and real-valued $U \in \mathbb{R}$ independent of $Y$, one has $\E_\oplus[\gamma_{\alpha,Y}(U)]=\gamma_{\alpha,\,\E_\oplus[Y]}\left(\E[U]\right)$.
\end{itemize}
\end{assumption}

Assumption \ref{Ass:p-map} provides structural properties of the \f mean that parallel key features of expectation in Euclidean spaces. Part (i) serves as the non-Euclidean analogue of the law of iterated expectations, ensuring coherence between conditional and marginal \f means. Part (ii) formalizes how \f means behave under geodesic interpolation, allowing weighted geodesics to be averaged in a manner consistent with their Euclidean counterparts. We verify that Assumption \ref{Ass:p-map} holds for Examples~\ref{exm:net}--\ref{exm:mea} in Appendix B.

Motivated by the Euclidean formulation, define
\begin{equation}\label{eq:dr}
    \Theta_t^{(\mathrm{DR})} = \E_\oplus\left[\gamma_{\mu_t(X),Y}\left(\frac{tT}{e(X)} + \frac{(1-t)(1-T)}{1-e(X)}\right)\right], \quad t \in \{0,1\}.
\end{equation}
By definition of the \f mean, one has $\Theta_t^{(\mathrm{DR})} = \argmin_{\nu \in \mathcal{M}}Q_t(\nu;e, \mu_t)$ where
\begin{equation}\label{eq:qt}
Q_t(\nu;e, \mu_t)=\E\left[d^2\left(\nu,\gamma_{\mu_t(X),Y}\left(\frac{tT}{e(X)} + \frac{(1-t)(1-T)}{1-e(X)}\right)\right)\right]. 
\end{equation}
The following result establishes the doubly robust identification of the GATE.

\begin{proposition}\label{Prp:GDR-identify}
Suppose that Assumptions~\ref{Ass:GATE} and \ref{Ass:p-map} hold. If either the propensity score model $e(\cdot)$ or the outcome regression models $\mu_t(\cdot)$ are correctly specified, that is, either $e(\cdot)=p(\cdot)$ or $\mu_t(\cdot)=m_t(\cdot)$ for $t \in \{0,1\}$, then
\begin{equation}\label{eq:GDR}
\gamma_{\E_{\oplus}[Y(0)], \E_{\oplus}[Y(1)]} = \gamma_{\Theta_0^{(\mathrm{DR})}, \Theta_1^{(\mathrm{DR})}}.
\end{equation}
\end{proposition}

Proposition~\ref{Prp:GDR-identify} shows that the GATE remains identifiable when either nuisance component is correctly specified, thereby preserving the double robustness property. This provides a direct geometric analogue of classical doubly robust estimands, where Euclidean differences and expectations are replaced by geodesics and \f means, respectively.

\subsection{Doubly robust estimation}\label{subsec:GATE2}

Building on Proposition~\ref{Prp:GDR-identify}, we construct a sample-based doubly robust estimator for the GATE. Specifically, the estimator is defined as the geodesic $\gamma_{\hat{\Theta}_0^{(\mathrm{DR})},\hat{\Theta}_1^{(\mathrm{DR})}}$, where 
\begin{align}
    \hat{\Theta}_t^{(\mathrm{DR})}&=\argmin_{\nu \in \mathcal{M}}Q_{n,t}(\nu;\hat{e}, \hat{\mu}_t),\label{eq:drhat}\\
    Q_{n,t}(\nu;\hat{e}, \hat{\mu}_t)&= \frac{1}{n}\sum_{i=1}^nd^2\left(\nu, \gamma_{\hat{\mu}_t(X_i), Y_i}\left(\frac{tT_i}{\hat{e}(X_i)} + \frac{(1-t)(1 - T_i)}{1 - \hat{e}(X_i)}\right)\right).\label{eq:qnt}
\end{align}
Here $\hat{e}(\cdot)$ and $\hat{\mu}_t(\cdot)$ denote estimators of the propensity score model $e(\cdot)$ and the outcome regression models $\mu_t(\cdot)$, respectively. The objective function $Q_{n,t}$ corresponds to the empirical \f functional of the doubly robust pseudo-outcomes, and thus $\hat{\Theta}_t^{(\mathrm{DR})}$ can be interpreted as their sample \f mean.

In practice, the choice of estimators for the nuisance components depends on the application. In our real data analysis, the propensity score is estimated via logistic regression, while the outcome regression function is estimated using global \f regression \citep{PeMu19}. Specifically, for $t \in \{0,1\}$,
\begin{equation}\label{eq:GFR}
    \mu_t(x) =\argmin_{\nu \in \mathcal{M}}\E\left[\left\{1 + (X - \E[X])'\Sigma^{-1}(x - \E[X])\right\}d^2(\nu,Y)|T=t\right],
\end{equation}
where $\Sigma = \Var(X)$. The corresponding sample estimator is
\[\hat{\mu}_t(x)=\argmin_{\nu \in \mathcal{M}}\frac{1}{N_t}\sum_{i \in I_t} \left\{1 + (X_i - \bar{X})'\hat{\Sigma}^{-1}(x - \bar{X})\right\}d^2(\nu,Y_i),\]
where $I_t = \{1 \leq i \leq n : T_i = t\}$, $N_t = |I_t|$, $\bar{X} = n^{-1} \sum_{i=1}^n X_i$, and $\hat{\Sigma} = n^{-1} \sum_{i=1}^n (X_i - \bar{X})(X_i - \bar{X})'$. The \f regression framework extends classical regression methods to accommodate responses taking values in general metric spaces; see \citet{PeMu19} for further details.

\begin{remark}[Cross-fitting estimator]
In the Euclidean case, doubly robust estimation methods have been combined with the cross-fitting approach to mitigate over-fitting \citep{CCDDHNR18}. The proposed doubly robust estimator $\gamma_{\hat{\Theta}_0^{(\mathrm{DR})},\hat{\Theta}_1^{(\mathrm{DR})}}$ can be similarly adapted to incorporate cross-fitting. Let $\{S_k\}_{k=1}^K$ be a (random) partition of $\{1,\dots,n\}$ into $K$ subsets. For each $k$, let $\hat{e}_k$ and $\hat{\mu}_{t,k}$ be the estimators of $e$ and $\mu_t(x)$ computed from the data $\{Y_i,T_i,X_i\}_{i \in S_{-k}}$, where $S_{-k} = \cup_{j \neq k}S_j$. Setting 
\begin{align*}
\hat{\Theta}_{t,k}^{(\mathrm{DR})}& = \argmin_{\nu \in \mathcal{M}}\frac{1}{n_k}\sum_{i \in S_k}d^2\left(\nu, \gamma_{\hat{\mu}_{t,k}(X_i), Y_i}\left(\frac{tT_i}{\hat{e}_k(X_i)} + \frac{(1-t)(1 - T_i)}{1 - \hat{e}_k(X_i) }\right)\right),\quad t \in \{0,1\},
\end{align*}
where $n_k$ is the sample size of $S_k$, the cross-fitting doubly robust estimator for the GATE is obtained as $\gamma_{\hat{\Theta}_0^{(\mathrm{CF})},\hat{\Theta}_1^{(\mathrm{CF})}}$, with 
\begin{align}\label{eq:cf}
\hat{\Theta}_t^{(\mathrm{CF})}& = \argmin_{\nu \in \mathcal{M}}\sum_{k=1}^K\frac{n_k}{n}d^2(\nu, \hat{\Theta}_{t,k}^{(\mathrm{DR})}),\quad t \in \{0,1\}.
\end{align}
\end{remark}

\begin{remark}[Outcome regression and inverse probability weighting estimators]\label{rem:OR-IPW}
When the outcome regression models are correctly specified, that is, $\mu_t(\cdot)=m_t(\cdot)$ for $t\in\{0,1\}$, the \f mean of the potential outcome satisfies $\E_\oplus[Y(t)] = \E_\oplus[\mu_t(X)]$. This motivates the outcome regression estimator of the GATE, obtained by replacing population quantities with their empirical counterparts. Specifically, the estimator is given by $\gamma_{\hat{\Theta}_0^{(\mathrm{OR})},\hat{\Theta}_1^{(\mathrm{OR})}}$, where
\begin{align*}
\hat{\Theta}_t^{(\mathrm{OR})} = \argmin_{\nu \in \mathcal{M}}\frac{1}{n}\sum_{i=1}^{n}d^2(\nu,\hat{\mu}_t(X_i)),\quad t \in \{0,1\}.
\end{align*}

When the propensity score model is correctly specified, that is, $e(\cdot)=p(\cdot)$, the mean potential outcome can be identified via inverse probability weighting. In the Euclidean setting, this relies on the identity
\[\E[Y(t)]=\E\left[\left(\frac{tT}{e(X)} + \frac{(1-t)(1-T)}{1-e(X)}\right)Y\right],\]
where the weights re-balance the observed outcomes to mimic a randomized experiment. Rewriting this expression around a reference point $y_\oplus$,
\[\E[Y(t)]=\E\left[y_\oplus+\left(\frac{tT}{e(X)} + \frac{(1-t)(1-T)}{1-e(X)}\right)(Y-y_\oplus)\right],\]
suggests a geometric interpretation: the pseudo-outcome is obtained by transporting $y_\oplus$ toward $Y$ along the line segment connecting them, with a weight determined by the propensity score. In a geodesic metric space, this linear interpolation is naturally replaced by a geodesic. Accordingly, we define
\[\Theta_t^{(\mathrm{IPW})}=\E_{\oplus}\left[\gamma_{y_\oplus, Y}\left(\frac{tT}{e(X)} + \frac{(1-t)(1-T)}{1-e(X)}\right)\right],\]
where $y_\oplus$ is a reference point, taken in practice as the sample \f mean of $\{Y_i\}_{i=1}^n$. The inverse probability weighting estimator of the GATE is then $\gamma_{\Theta_0^{(\mathrm{IPW})}, \Theta_1^{(\mathrm{IPW})}}$, with empirical version $\gamma_{\hat{\Theta}_0^{(\mathrm{IPW})}, \hat{\Theta}_1^{(\mathrm{IPW})}}$, where
\begin{align*}
\hat{\Theta}_t^{(\mathrm{IPW})} = \argmin_{\nu \in \mathcal{M}}\frac{1}{n}\sum_{i=1}^{n}d^2\left(\nu,\gamma_{y_\oplus,Y_i}\left(\frac{tT_i}{\hat{e}(X_i)} + \frac{(1-t)(1 - T_i)}{1 - \hat{e}(X_i) }\right)\right),\quad t \in \{0,1\},
\end{align*}
The asymptotic properties of $\hat{\Theta}_t^{(\mathrm{OR})}$ and $\hat{\Theta}_t^{(\mathrm{IPW})}$ are provided in Appendix A and their finite-sample performance is examined in Section~\ref{sec:sim}. 
\end{remark}

\section{Consistency and rates of convergence}\label{sec:DR}

To study the asymptotic behavior of the proposed doubly robust and cross-fitting estimators, we impose conditions on the nuisance estimators for the propensity score and outcome regression functions.

\begin{assumption}\label{Ass:model-prop-reg} \quad
\begin{itemize}
\item[(i)] $\sup_{x \in \mathcal{X}}|\hat{e}(x) - e(x)| = O_{\Prob}(\varrho_n)$ with $\varrho_n \to 0$ as $n \to \infty$, and $\eta_0 \leq e(x) \leq 1-\eta_0$ for all $x \in \mathcal{X}$, where $\eta_0$ is the constant in Assumption~\ref{Ass:GATE}(i).
\item[(ii)] For $t \in \{0,1\}$, $\sup_{x \in \mathcal{X}}d(\hat{\mu}_t(x),\mu_t(x)) = O_{\Prob}(r_n)$ with $r_n \to 0$ as $n \to \infty$. 
\end{itemize}
\end{assumption}

Assumption~\ref{Ass:model-prop-reg} characterizes the convergence rates of the nuisance estimators $\hat e$ and $\hat\mu_t$ approach their respective population targets $e$ and $\mu_t$. These population targets are determined by the chosen models and need not coincide with the true propensity score $p$ or the true conditional \f mean $m_t$ unless the corresponding models are correctly specified. For commonly used parametric models, such as logistic regression for the propensity score, Assumption~\ref{Ass:model-prop-reg}(i) typically holds with $\varrho_n = n^{-1/2}$. For the outcome regression, if one employs global \f regression as in \eqref{eq:GFR}, then under suitable regularity conditions, Assumption~\ref{Ass:model-prop-reg}(ii) holds with $r_n = n^{-\alpha'}$ for some $\alpha' < 1/2$; see Theorem~2 of \citet{PeMu19}. Alternatively, local \f regression \citep{ChMu22} may be used to allow for greater flexibility, at the cost of slower convergence rates and sensitivity to the dimensionality of the covariates.

In addition to the regularity conditions on the nuisance estimators, we impose a geometric condition on the underlying metric space.

\begin{assumption}\label{Ass:geodesic-dist}
For all $\alpha_1,\alpha_2 \in (\mathcal{M},d)$, 
\begin{equation}\label{ineq:geodesic-dist}
\sup_{\beta \in \mathcal{M}, \kappa \in [1/(1-\eta_0),1/\eta_0]}d\left(\gamma_{\alpha_1,\beta}(\kappa), \gamma_{\alpha_2,\beta}(\kappa)\right) \leq C_0d(\alpha_1,\alpha_2),
\end{equation}
for some positive constant $C_0$ depending only on $\eta_0$. 
\end{assumption}

Assumption~\ref{Ass:geodesic-dist} provides a uniform Lipschitz control on geodesic interpolations with respect to their starting points, which is essential for analyzing the empirical objective function in \eqref{eq:qnt}. When the space $\mathcal{M}$ is a CAT($0$) space (i.e., a geodesic metric space that is globally non-positively curved in the sense of Alexandrov) \citep{stur:03}, this assumption is automatically satisfied. Many metric spaces of statistical interest are CAT($0$) spaces, such as the Wasserstein space of univariate distributions, the space of networks expressed as graph Laplacians with the (power) Frobenius metric, SPD matrices with the affine-invariant metric \citep{than:23}, the Log-Cholesky metric \citep{lin:19:1} and various other metrics \citep{pigo:14}, as well as phylogenetic tree space with the BHV metric \citep{bill:01}. Assumption \ref{Ass:geodesic-dist} also holds for some relevant positively curved spaces such as open hemispheres or orthants of spheres that are relevant for compositional data with the Riemannian metric; these spaces are not CAT(0), see Example \ref{exm:com}.

We further impose conditions on the population and empirical objective functions, as well as the identification of the target parameter.
 
\begin{assumption}\label{Ass:DR-criterion} \quad
\begin{itemize}
\item[(i)] For $t \in \{0,1\}$, the population minimizer $\Theta_t^{(\mathrm{DR})}$ and its empirical counterpart $\hat{\Theta}_t^{(\mathrm{DR})}$, defined in~\eqref{eq:dr} and~\eqref{eq:drhat}, exist and are unique. Moreover, for any $\varepsilon>0$, 
\[\inf_{d(\nu, \Theta_t^{(\mathrm{DR})})>\varepsilon}Q_t(\nu;e,\mu_t) > Q_t(\Theta_t^{(\mathrm{DR})};e,\mu_t).\]
\item[(ii)] At least one of the nuisance models is correctly specified, that is, either $e(\cdot) = p(\cdot)$ or $\mu_t(\cdot) = m_t(\cdot)$ for $t \in \{0,1\}$.
\end{itemize}
\end{assumption}

Assumption~\ref{Ass:DR-criterion}(i) is a standard separation condition ensuring well-posedness and consistency of the M-estimator; see, for example, Chapter~3.2 of \citet{vaWe23}. Assumption~\ref{Ass:DR-criterion}(ii) is the familiar double robustness condition, which guarantees identification of the target parameter when at least one of the nuisance models is correctly specified. 

The following result provides the consistency of the doubly robust and cross-fitting estimators $\hat{\Theta}_t^{(\mathrm{DR})}$ in \eqref{eq:drhat} and $\hat{\Theta}_t^{(\mathrm{CF})}$ in \eqref{eq:cf}.

\begin{theorem}\label{Thm:GDR-conv}
Suppose that Assumptions \ref{Ass:GATE}, \ref{Ass:p-map}, \ref{Ass:geodesic-dist}, and \ref{Ass:DR-criterion} hold.
\begin{itemize}
\item[(i)]Under Assumption \ref{Ass:model-prop-reg}, as $n \to \infty$, 
\[d(\hat{\Theta}_t^{(\mathrm{DR})},\E_\oplus[Y(t)]) = o_{\Prob}(1),\quad t \in \{0,1\}.\]
\item[(ii)] Suppose that Assumption \ref{Ass:model-prop-reg} holds for $\{\hat{e}_k,\hat{\mu}_{t,k}\}_{k=1}^K$ and there exist constants $c_1$ and $c_2$ such that $0<c_1 \leq n_k/n \leq c_2<1$ for all $k=1,\dots,K$. Then as $n \to \infty$, 
\[d(\hat{\Theta}_t^{(\mathrm{CF})},\E_\oplus[Y(t)]) = o_{\Prob}(1),\quad t \in \{0,1\}.\]
\end{itemize}
\end{theorem}

To obtain rates of convergence for estimators $\hat{\Theta}_t^{(\mathrm{DR})}$ and $\hat{\Theta}_t^{(\mathrm{CF})}$, we need additional entropy conditions to quantify the complexity of the space. Let $B_\delta(\omega)$ be the ball of radius $\delta$ centered at $\omega \in \mathcal{M}$ and $N(\varepsilon, B_\delta(\omega),d)$ its covering number using balls of size $\varepsilon$.

\begin{assumption}\label{Ass:GDR-rate}
For $t \in \{0,1\}$, 
\begin{itemize}
\item[(i)] As $\delta \to 0$, 
\begin{align}
J_t(\delta)&:=\int_0^1 \sqrt{1 + \log N(\delta \varepsilon, B_\delta(\Theta_t^{(\mathrm{DR})}),d)}d\varepsilon = O(1), \label{eq:coverN-bound1}\\
J_{\mu_t}(\delta) &:= \int_0^1 \sqrt{1 + \log N(\delta \varepsilon, B_{\delta'_1}(\mu_t),d_{\infty})}d\varepsilon = O(\delta^{-\vartheta}), \label{eq:coverN-bound2}
\end{align}
for some $\delta'_1>0$ and $\vartheta \in (0,1)$, where for $\nu,\mu:\mathcal{X} \to \mathcal{M}$, $d_\infty(\nu,\mu) = \sup_{x \in \mathcal{X}}d(\nu(x),\mu(x))$.
\item[(ii)] There exist positive constants $\eta, \eta_1, C, C_1$, and $\beta>1$ such that
\begin{align}\label{eq:exp}
\inf_{\substack{\sup_{x\in\mathcal{X}}d(\tilde{\mu}(x),\mu_t(x))\leq \eta_1\\ \sup_{x\in\mathcal{X}}|\tilde{e}(x)-e(x)| \leq \eta_1}}\hspace{0.2em}
\inf_{d(\nu,\Theta_t^{(\mathrm{DR})})<\eta} 
\Bigg\{& Q_t(\nu; \tilde{e}, \tilde{\mu}) - Q_t(\Theta_t^{(\mathrm{DR})}; \tilde{e}, \tilde{\mu}) - \nonumber \\
&Cd(\nu, \Theta_t^{(\mathrm{DR})})^\beta
+ C_1\eta_1^{\frac{\beta}{2(\beta-1)}}d(\nu, \Theta_t^{(\mathrm{DR})})^{\frac{\beta}{2}} 
\Bigg\} \geq 0.
\end{align}
\end{itemize}
\end{assumption}

These assumptions control the behavior of the (centered) objective function $Q_{n,t}-Q_t$ around the minimum. Assumption \ref{Ass:GDR-rate}(i) postulates an entropy condition for $\Theta_t^{(\mathrm{DR})}$ and $\mu_t$. A sufficient condition for (\ref{eq:coverN-bound2}) is $J'_{\mu_t}(\delta) :=\int_0^1 \sup_{x \in \mathcal{X}}\sqrt{1 + \log N(\bar{c}\delta\varepsilon, B_{\delta'_1}(\mu_t(x)),d)}d\varepsilon = O(\delta^{-\vartheta'})$ for some positive constants $\bar{c}$ and $\vartheta' \in (0,1)$. One can verify (\ref{eq:coverN-bound1}) and
\begin{align}\label{eq:covering-reg}
J'_{\mu_t}(\delta) = O(-\log \delta),\ \text{as $\delta \to 0$}
\end{align}
for common metric spaces including Examples \ref{exm:net}--\ref{exm:com}, for which therefore Assumption \ref{Ass:GDR-rate}(i) is satisfied. If (\ref{eq:covering-reg}) holds, we can take $\vartheta \in (0,1)$ in Theorem \ref{Thm:GDR-rate} below arbitrarily small.

Assumption~\ref{Ass:GDR-rate}(ii) imposes a curvature condition on the population objective function $Q_t$ in a neighborhood of its minimizer, ensuring sufficient separation. It extends Assumption~B3 in \citet{DeKe20} for M-estimators with plug-in nuisance components. When $Q_t$ does not involve nuisance parameters, the fourth term in (\ref{eq:exp}) vanishes, in analogy to Assumption~(U2) in \citet{PeMu19}. This condition is satisfied for a broad class of metric spaces; in particular, Examples~\ref{exm:net}--\ref{exm:com} satisfy it with $\beta = 2$; see Propositions~B.1--B.5 in Appendix B. In the Euclidean case, $\beta = 2$, corresponding to a quadratic local approximation of $Q_t$, in which the third and fourth terms in~\eqref{eq:exp} capture the quadratic and linear components, respectively.

The following result establishes the convergence rates for the proposed doubly robust and cross-fitting estimators.

\begin{theorem}\label{Thm:GDR-rate}
Suppose that Assumptions \ref{Ass:GATE}, \ref{Ass:p-map}, \ref{Ass:geodesic-dist}, and \ref{Ass:DR-criterion} hold. 
\begin{itemize}
\item[(i)] If Assumptions \ref{Ass:model-prop-reg} and \ref{Ass:GDR-rate} hold, then for each $\epsilon \in (0,1)$, as $n \to \infty$, with $\varrho_n$ and $r_n$ as in Assumption~\ref{Ass:model-prop-reg} and $\beta, \vartheta$ as in Assumption~\ref{Ass:GDR-rate}, 
\begin{align}
d(\hat{\Theta}_t^{(\mathrm{DR})},\E_\oplus[Y(t)]) &= O_{\Prob}\left( n^{-\frac{1}{2(\beta - 1 +\vartheta)}} + (\varrho_n + r_n)^{\frac{1-\epsilon}{\beta-1}}\right),\quad t \in \{0,1\}. \label{eq:GDR-rate} 
\end{align}
\item[(ii)] If Assumptions \ref{Ass:model-prop-reg} and \ref{Ass:GDR-rate} hold for $\{\hat{e}_k,\hat{\mu}_{t,k}\}_{k=1}^K$ and there exist constants $c_1$ and $c_2$ such that $0<c_1 \leq n_k/n \leq c_2<1$ for all $k=1,\dots,K$, then for each $\epsilon \in (0,1)$, as $n \to \infty$, 
\begin{align}
d(\hat{\Theta}_t^{(\mathrm{CF})},\E_\oplus[Y(t)]) &= O_{\Prob}\left( n^{-\frac{1}{2(\beta - 1 +\vartheta)}} + (\varrho_n + r_n)^{\frac{1-\epsilon}{\beta-1}}\right),\quad t \in \{0,1\}. \label{eq:GDR-rate-CF} 
\end{align}
\end{itemize}
\end{theorem}

Both the doubly robust and cross-fitting estimators achieve the same convergence rate. The first term reflects the intrinsic stochastic fluctuation of the estimator when the nuisance components are treated as fixed, and reduces to the usual rate for \f means when $\vartheta=0$. The second term arises from the estimation of the nuisance functions and depends on the rates at which $\hat e$ and $\hat\mu_t$ converge to their respective population targets $e$ and $\mu_t$, which need not coincide with the true propensity score $p$ and conditional \f mean $m_t$ unless the corresponding models are correctly specified. In a typical scenario, we have $\beta=2$, $\varrho_n=n^{-1/2}$, and $r_n=n^{-\alpha_1}$ with any $\alpha_1<1/2$ so that the rate becomes $O_{\Prob}( n^{-\frac{1}{2(1+\vartheta)}} + n^{-\alpha_1(1-\epsilon)})$ for any $\vartheta, \epsilon \in (0,1)$, which can be arbitrarily close to the parametric rate.

\begin{remark}[Convergence rates of outcome regression and inverse probability weighting estimators]
In Appendix A, we provide convergence rates of estimators $\hat{\Theta}_t^{(\mathrm{OR})}$ and $\hat{\Theta}_t^{(\mathrm{IPW})}$ that rely on the correct specification of the outcome regression function and the propensity score function, respectively. Under some regularity conditions one can show that for each $\epsilon \in (0,1)$, as $n \to \infty$, 
\begin{align*}
d(\hat{\Theta}_t^{(\mathrm{OR})},\E_\oplus[Y(t)]) &= O_{\Prob}\left(n^{-\frac{1}{2(\beta-1+\vartheta_1)}} + r_{n,1}^{\frac{1-\epsilon}{\beta-1}}\right),\\
d(\hat{\Theta}_t^{(\mathrm{IPW})},\E_\oplus[Y(t)]) &= O_{\Prob}\left(n^{-\frac{1}{2(\beta-1+\vartheta_2)}} + \varrho_{n,1}^{\frac{1-\epsilon}{\beta-1}}\right),
\end{align*}
where $\vartheta_1, \vartheta_2 \in (0,1)$ are constants arising from entropy conditions for the models of the outcome regression function and the propensity score, corresponding to $\vartheta$ in (\ref{eq:coverN-bound2}). Here $r_{n,1}$ and $\varrho_{n,1}$ are null sequences that correspond to the convergence rates of $\hat{\mu}_t(\cdot)$ and $\hat{e}(\cdot)$, respectively, i.e.,
\[\sup_{x \in \mathcal{X}}d(\hat{\mu}_t(x),m_t(x)) = O_{\Prob}(r_{n,1}),\ \sup_{x \in \mathcal{X}}|\hat{e}(x) - p(x)| = O_{\Prob}(\varrho_{n,1}).\]
\end{remark}

\section{Uncertainty quantification}\label{sec:uq}

In general metric spaces, inference is challenging due to the absence of linear structure, which precludes the direct application of classical asymptotic tools such as central limit theorems for estimators. To overcome this difficulty, we develop an intrinsic inference framework based on the Fr\'echet objective functions that characterize the target parameters. These objective functions are real-valued and therefore amenable to standard probabilistic analysis. The key idea is that the population Fr\'echet objective function is uniquely minimized at $\Eo[Y(t)]$. This enables us to assess whether a candidate point $\omega \in \mathcal{M}$ coincides with $\Eo[Y(t)]$ through the behavior of the corresponding objective function. We first develop objective function-based test statistics, then construct confidence regions, and finally propose a test for no treatment effect.

\subsection{Objective function-based test statistics}\label{subsec:obj}
Suppose that the assumptions of Theorem~\ref{Thm:GDR-rate} hold. Let $\omega_0 \in \mathcal{M}$ be fixed. Let $m=m_n \to \infty$ with $m/n \to 0$ as $n\to\infty$. Define
\begin{align*}
\check{T}_m^{(t)}(\omega_0) &= \frac{1}{\sqrt{m}}\sum_{i=1}^m\left\{d^2\left(\omega_0, \gamma_{\mu_t(X_i), Y_i}\left(\frac{tT_i}{e(X_i)} + \frac{(1-t)(1 - T_i)}{1 - e(X_i)}\right)\right)-Q_t(\omega_0;e,\mu_t)\right\}.
\end{align*}
For notational simplicity we write the sum over \(i=1,\ldots,m\). This may be interpreted as summing over an arbitrary subset of \(m\) observations; by exchangeability, the choice of the subset is immaterial. Let $\leadsto$ and $\stackrel{\Prob}{\to}$ denote weak convergence and convergence in $\Prob$-probability, respectively. The central limit theorem yields 
\[\check{T}_m^{(t)}(\omega_0) \leadsto N(0,V^{(t)}(\omega_0)),\]
as $n \to \infty$, where
\[V^{(t)}(\omega_0) = \mathrm{Var}\left(d^2\left(\omega_0, \gamma_{\mu_t(X_i), Y_i}\left(\frac{tT_i}{e(X_i)} + \frac{(1-t)(1 - T_i)}{1 - e(X_i)}\right)\right)\right).\]

Since $\Eo[Y(t)]$ uniquely minimizes $Q_t(\cdot;e,\mu_t)$,
\[B_t(\omega_0,\Eo[Y(t)]):=Q_t(\omega_0;e,\mu_t)-Q_t(\Eo[Y(t)];e,\mu_t)=\begin{cases}0,&\omega_0=\Eo[Y(t)],\\>0,&\omega_0\neq\Eo[Y(t)].\end{cases}\]
Letting $\tilde{T}_m^{(t)}(\omega_0) = \check{T}_m^{(t)}(\omega_0) + \sqrt{m}B_t(\omega_0,\Eo[Y(t)])$, it follows that
\[\left|\tilde{T}_m^{(t)}(\omega_0)\right|\begin{cases}\leadsto |N(0,V^{(t)}(\omega_0))|,&\omega_0=\Eo[Y(t)],\\\stackrel{\Prob}{\to}\infty,&\omega_0\neq\Eo[Y(t)].\end{cases}\]

In practice, the nuisance functions are unknown. For $\mathrm{M}\in\{\mathrm{DR},\mathrm{CF}\}$, our plug-in statistic is defined as
\begin{align*}
T_{m,n}^{(t)}(\omega_0) &=\frac{1}{\sqrt{m}}\!\sum_{i=1}^m\!\left\{\!d^2\!\!\left(\omega_0, \gamma_{\hat{\mu}_t(X_i), Y_i}\!\!\left(\!\frac{tT_i}{\hat{e}(X_i)} \!+\! \frac{(1-t)(1 - T_i)}{1 - \hat{e}(X_i)}\!\right)\!\!\right)\!-Q_{n,t}(\hat{\Theta}_t^{(\mathrm{M})};\hat{e}, \hat{\mu}_t)\!\right\}, 
\end{align*}
where $\hat{\Theta}_t^{(\mathrm{M})}$, $\hat{e}$, and $\hat{\mu}_t$ are estimators constructed from the full sample $\mathcal{D}_n=\{(Y_i,T_i,X_i)\}_{i=1}^n$. The use of a sum over $m$ observations, with $m \to \infty$ and $m/n \to 0$, serves to separate the stochastic fluctuation of the statistic from the plug-in errors induced by estimating the nuisance functions. In particular, the statistic fluctuates on the $\sqrt{m}$ scale, while the estimation errors of $\hat e$, $\hat \mu_t$, and $\hat \Theta_t^{(\mathrm{M})}$ are of smaller order under the conditions of Proposition~\ref{prp:test-stat-lim}, and therefore become asymptotically negligible. We formalize this observation in the following proposition.

\begin{proposition}\label{prp:test-stat-lim}
Suppose that the assumptions of Theorem \ref{Thm:GDR-rate} hold. Additionally, assume that $m^{-1}+\sqrt{m}\{n^{-\frac{1}{2}\min\{\frac{1}{(\beta - 1 +\vartheta)},1\}} + (\varrho_n + r_n)^{\min\{\frac{1-\epsilon}{\beta-1},1\}}\}\to 0$ as $n \to \infty$. Then for $t,t'\in\{0, 1\}$ and $\mathrm{M} \in \{\mathrm{DR},\mathrm{CF}\}$, 
\[T_{m,n}^{(t)}(\hat{\Theta}_{t'}^{(\mathrm{M})}) = T_{m,n}^{(t)}(\Eo[Y(t')]) +o_{\Prob}(1) = \tilde{T}_m^{(t)}(\Eo[Y(t')]) + o_{\Prob}(1)\]
and hence 
\begin{align*}
|T_{m,n}^{(t)}(\hat{\Theta}_{t'}^{(\mathrm{M})})| 
\begin{cases}
\leadsto |N(0,V^{(t)}(\Eo[Y(t')]))|,&t=t',\\
\stackrel{\Prob}{\to} \infty,&t\neq t'.
\end{cases}
\end{align*}
\end{proposition}

\subsection{Confidence regions}
Proposition~\ref{prp:test-stat-lim} enables the construction of confidence regions for $\Eo[Y(t)]$.  For $\alpha \in (0,1)$, let $q_{1-\alpha}^{(t)}$ denote the $(1-\alpha)$-quantile of $|N(0,V^{(t)}(\Eo[Y(t)]))|$. Since $V^{(t)}(\Eo[Y(t)])$ is unknown, we estimate this quantile via the following multiplier bootstrap procedure. Let $\{\xi_i\}_{i=1}^m$ be i.i.d. standard normal random variables defined on a probability space $(\Omega_\xi,\mathcal{A}_\xi,\Prob_\xi)$, independent of the data $\mathcal{D}_n$.

For $\mathrm{M} \in \{\mathrm{DR,\mathrm{CF}}\}$, define the bootstrap statistics
\begin{align*}
&T_{m,n}^{(t,*)}(\hat{\Theta}_t^{(\mathrm{M})})\\ 
&= \frac{1}{\sqrt{m}}\sum_{i=1}^m \xi_i\!\left\{d^2\!\left(\hat{\Theta}_t^{(\mathrm{M})}, \gamma_{\hat{\mu}_t(X_i), Y_i}\left(\frac{tT_i}{\hat{e}(X_i)} + \frac{(1-t)(1 - T_i)}{1 - \hat{e}(X_i)}\right)\right)-Q_{m,t}(\hat{\Theta}_t^{(\mathrm{M})};\hat{e}, \hat{\mu}_t)\right\}
\end{align*}
and let $\hat{q}_{1-\alpha}^{(t)}$ denote the $(1-\alpha)$-quantile of $|T_{m,n}^{(t,*)}(\hat{\Theta}_t^{(\mathrm{M})})|$ given the data $\mathcal{D}_n$. We define the $(1-\alpha)$ confidence region for $\Eo[Y(t)]$ as
\[\hat{R}_{1-\alpha}^{(t)} = \{\omega \in \mathcal{M}: |T_{m,n}^{(t)}(\omega)|\leq \hat{q}_{1-\alpha}^{(t)}\}.\]

\begin{proposition}\label{prp:GDR-boot}
    Under the same assumption in Proposition \ref{prp:test-stat-lim}, we have $\Prob(\Eo[Y(t)] \in \hat{R}_{1-\alpha}^{(t)}) \to 1-\alpha$ as $n \to \infty$.
\end{proposition}
The confidence region $\hat{R}_{1-\alpha}^{(t)}$ admits a natural interpretation: it collects all points whose Fr\'echet objective values cannot be statistically distinguished from that of the true minimizer at level $\alpha$. In this sense, it provides an intrinsic characterization of uncertainty for $\Eo[Y(t)]$ that respects the geometry of the underlying space.

\subsection{Testing for no treatment effect}
We consider testing the null hypothesis of no treatment effect,
\[\mathbb{H}_0:d(\Eo[Y(0)],\Eo[Y(1)])=0\quad\text{vs.}\quad\mathbb{H}_1:d(\Eo[Y(0)],\Eo[Y(1)])>0.\]
The objective function-based statistics developed in Section~\ref{subsec:obj} provide a natural basis for testing $\mathbb{H}_0$. In particular, Proposition~\ref{prp:test-stat-lim} implies that, for $\mathrm{M}\in\{\mathrm{DR},\mathrm{CF}\}$,
\[T_{m,n}^{(0)}(\hat{\Theta}_1^{(\mathrm{M})})-\tilde{T}_m^{(0)}(\Eo[Y(1)])=o_{\Prob}(1).\]
Consequently,
\[\left|T_{m,n}^{(0)}(\hat{\Theta}_1^{(\mathrm{M})})\right|\begin{cases}\leadsto |N(0,V^{(0)}(\Eo[Y(0)]))|,&\text{under }\mathbb{H}_0,\\\stackrel{\Prob}{\to}\infty,&\text{under }\mathbb{H}_1.\end{cases}\]

This characterization suggests a natural testing procedure. Under $\mathbb{H}_0$, the estimator $\hat{\Theta}_1^{(\mathrm{M})}$ behaves asymptotically as if it were evaluated at the true minimizer of the objective function for $t=0$, and hence the corresponding statistic remains stochastically bounded. In contrast, under $\mathbb{H}_1$, the two population minimizers are distinct, and the objective function gap induces divergence of the statistic.

Using the confidence region $\hat{R}_{1-\alpha}^{(0)}$ defined in the previous subsection, we reject $\mathbb{H}_0$ whenever $\hat{\Theta}_1^{(\mathrm{M})}\notin \hat{R}_{1-\alpha}^{(0)}$. Equivalently, the rejection rule can be expressed as
\[\left|T_{m,n}^{(0)}(\hat{\Theta}_1^{(\mathrm{M})})\right|>\hat{q}_{1-\alpha}^{(0)}.\]
Under the assumptions of Proposition~\ref{prp:test-stat-lim}, this test has asymptotic size $\alpha$ and is consistent, in the sense that
\[\Prob(\hat{\Theta}_1^{(\mathrm{M})}\notin \hat{R}_{1-\alpha}^{(0)})\begin{cases}\to\alpha,&\text{under }\mathbb{H}_0,\\\to1,&\text{under }\mathbb{H}_1.\end{cases}.\]

\section{Simulation studies}\label{sec:sim} 
\subsection{Implementation details and simulation scenarios}
The algorithm for the proposed approach is outlined in Algorithm~\ref{alg:gci}. The global \f regression involved in the second step is implemented using the R package \texttt{frechet} \citep{chen:20}. The minimization problem one needs to solve requires specific considerations for various geodesic spaces. For the spaces in Examples \ref{exm:net}--\ref{exm:mea}, the minimization problem can be reduced to convex quadratic optimization \citep{stel:20}. For Example \ref{exm:com}, the necessary optimization can be performed using the trust regions algorithm \citep{geye:20}. The third step involves the calculation of \f means, which reduces to the entry-wise mean of matrices for Examples \ref{exm:net}, \ref{exm:cor}, and the point-wise mean of functions for Examples \ref{exm:fun}, \ref{exm:mea} due to the convexity of these spaces. The \f mean for Example \ref{exm:com} can be implemented using the R package \texttt{manifold} \citep{mull:21:2}.
\begin{algorithm}
	\KwIn{data $\{(Y_i, T_i, X_i)\}_{i=1}^n$.}
	\KwOut{doubly robust estimator of geodesic average treatment effect (GATE) $\gamma_{\hat{\Theta}_0^{(\text{DR})}, \hat{\Theta}_1^{(\text{DR})}}$.}
        $\hat{e}(\cdot)\longleftarrow$ the estimated propensity score using logistic regression with data $\{T_i, X_i\}_{i=1}^n$\;
        $\hat{\mu}_t(\cdot)\longleftarrow$ the estimated outcome regression function using global \f regression:
        \[\hat{\mu}_t(x)=\argmin_{\nu \in \mathcal{M}}\frac{1}{N_t}\sum_{i \in I_t} \{1 + (X_i - \bar{X})'\hat{\Sigma}^{-1}(x - \bar{X})\}d^2(\nu,Y_i),\quad t \in \{0,1\},\] where $I_t=\{1\leq i \leq n:T_i = t\}$, $N_t$ is the cardinality of $I_t$, $\bar{X} = n^{-1}\sum_{i=1}^n X_i$, and $\hat{\Sigma} = n^{-1}\sum_{i=1}^n(X_i - \bar{X})(X_i - \bar{X})'$\; 
	$\hat{\Theta}_t^{(\text{DR})}\longleftarrow$ the estimated mean potential outcomes:
        \[\hat{\Theta}_t^{\mathrm{(DR)}}=\argmin_{\nu\in\mathcal{M}}\frac{1}{n}\sum_{i=1}^nd^2\left(\nu, \gamma_{\hat{\mu}_t(X_i), Y_i}\left(\frac{tT_i}{\hat{e}(X_i)} + \frac{(1-t)(1 - T_i)}{1 - \hat{e}(X_i) }\right)\right)\]
        where $\hat{e}(\cdot)$ and $\hat{\mu}_t(\cdot)$ are the estimated propensity score and outcome regression function\;
        $\gamma_{\hat{\Theta}_0^{(\text{DR})}, \hat{\Theta}_1^{(\text{DR})}}\longleftarrow$ the doubly robust estimator of GATE.
	\caption{Geodesic Causal Inference}
	\label{alg:gci}
\end{algorithm}

To assess the performance of the proposed doubly robust estimators, we report here the results of simulations for various settings, specifically for the space of SPD matrices equipped with the Frobenius metric and the space of three-dimensional compositional data $\mathcal{S}_+^{2}$ equipped with the geodesic metric; see Examples \ref{exm:cor} and \ref{exm:com}.

We consider sample sizes $n=100, 300, 1000$, with 500 Monte Carlo runs. For the $q$th Monte Carlo run, with $\gamma_{\hat{\Theta}_0^q, \hat{\Theta}_1^q}$ denoting the GATE estimator, the average quality of the estimation over the 500 Monte Carlo runs is assessed by the average squared error (ASE)
\[\mathrm{ASE}=\frac{1}{500}\sum_{q=1}^{500}\{d^2(\hat{\Theta}_0^q, \E_\oplus[Y(0)])+d^2(\hat{\Theta}_1^q, \E_\oplus[Y(1)])\},\]
where $d$ is the Frobenius metric $d_F$ for SPD matrices and the geodesic metric $d_g$ for compositional data. 

In all simulations, the confounder $X$ follows a uniform distribution $[-1, 1]$. The treatment $T$ has a conditional Bernoulli distribution depending on $X$ with $\Prob(T=1|X)=\mathrm{expit}(0.75X)$ where $\mathrm{expit}(\cdot)=\mathrm{exp}(\cdot)/(1+\mathrm{exp}(\cdot))$. 
We consider two specifications for the outcome regression: a global \f regression model with predictor $X$ in $m_t(X)$ (correct specification) and one with the predictor $X^2$ (incorrect specification), and two specifications for the propensity score model: a logistic regression model with predictor $X$ (correct specification) and a model with predictor $X^2$ (incorrect specification). To demonstrate the double robustness of estimators $\hat{\Theta}_t^{\text{(DR)}}$ and $\hat{\Theta}_t^{\text{(CF)}}$, we compare them with outcome regression (OR) and inverse probability weighting (IPW) estimators.

\subsection{SPD matrices}
Random SPD matrices $Y$ are generated as follows. First, the lower-triangular entries of $Y$ are independently generated as $Y_{jk}=T+X+2+\epsilon$, where $1 \le j,k \le m,\, j>k$ and $\epsilon$ is independently generated from a uniform distribution on $[-0.1, 0.1]$. Due to symmetry, the upper-triangular entries of $Y$ are then $Y_{kj}=Y_{jk}$ for $1 \le j,k \le m, \ j>k$. To ensure positive semi-definiteness, the diagonal entries $Y_{jj}$ are set to equal the sum of all off-diagonal entries in the same row, $\sum_{k\neq j}Y_{jk}$. This construction results in a diagonally dominated matrix, thus guaranteeing positive semi-definiteness. 

We aim to estimate the GATE, whose true value is $\gamma_{\E_\oplus[Y(0)], \E_\oplus[Y(1)]}$ where $(\E_\oplus[Y(0)])_{jk}=2, (\E_\oplus[Y(0)])_{jj}=18$, and $(\E_\oplus[Y(1)])_{jk}=3, (\E_\oplus[Y(1)])_{jj}=27$ for $1\leq j\neq k\leq m$. Table~\ref{tab:cov} presents the simulation results for $10 \times 10$ matrices, i.e., for $m=10$. The ASE of the doubly robust estimator decreases as the sample size increases when either the OR or IPW model is correctly specified, confirming the double robustness property. However, neither the OR nor IPW estimator demonstrates double robustness; their ASEs can be large even with a sample size of 1000 when the corresponding model is misspecified.

\begin{table}[tb]
	\caption{Simulation results for SPD matrices. Average squared errors and standard deviations (in parentheses) for GATE using four different estimation procedures. The model specifications are in the first two columns. DR: doubly robust; CF: cross-fitting; OR: outcome regression; IPW: inverse probability weighting.}
	\centering
	\begin{tabular}{lllllll}
	\hline
	\multicolumn{2}{c}{Model}&Sample size&\multicolumn{4}{c}{Estimator}
	\\\cline{1-2}\cline{4-7}
	OR&IPW&&DR&CF&OR&IPW
	\\\hline
	\multirow{3}{*}{\ding{51}}&\multirow{3}{*}{\ding{51}}&100&6.443 (9.580)&6.446 (9.578)&6.444 (9.580)&6.800 (9.829)
	\\
        &&300&1.988 (2.602)&1.989 (2.603)&1.988 (2.602)&2.127 (2.774)
        \\
        &&1000&0.560 (0.831)&0.560 (0.832)&0.560 (0.831)&0.588 (0.915)
        \\\hline
        \multirow{3}{*}{\ding{51}}&\multirow{3}{*}{\ding{55}}&100&6.444 (9.584)&6.445 (9.576)&6.444 (9.580)&33.67 (25.44)
	\\
        &&300&1.988 (2.602)&1.988 (2.602)&1.988 (2.602)&26.93 (13.38)
        \\
        &&1000&0.560 (0.831)&0.560 (0.831)&0.560 (0.831)&23.17 (7.39)
        \\\hline
	\multirow{3}{*}{\ding{55}}&\multirow{3}{*}{\ding{51}}&100&9.871 (14.26)&12.64 (19.44)&39.74 (27.90)&6.800 (9.829)
	\\
        &&300&3.271 (4.082)&3.424 (4.396)&31.43 (14.66)&2.127 (2.774)
        \\
        &&1000&0.935 (1.359)&0.938 (1.388)&27.63 (7.924)&0.588 (0.915)
	\\\hline
	\end{tabular}
	\label{tab:cov}
\end{table}

\subsection{Compositional data}
Writing $\phi=\pi(X+2)/8\in[\pi/8, 3\pi/8]$, we model the true regression functions $m_0(\cdot)$ and $m_1(\cdot)$ as
\[m_0(X)=(\cos(\phi), \frac{1}{2}\sin(\phi), \frac{\sqrt{3}}{2}\sin(\phi)),\quad m_1(X)=(\cos(\phi), \frac{\sqrt{3}}{2}\sin(\phi), \frac{1}{2}\sin(\phi)),\]
which are illustrated in Figure~\ref{fig:comp} as red and blue lines, respectively. The random outcome $Y$ on $\mathcal{S}_+^2$ for $T=0$ is then generated by adding a small perturbation to the true regression function. To this end, we first construct an orthonormal basis $(e_1, e_2)$ for the tangent space on $m_0(X)$ where 
\[e_1=(\sin(\phi), -\frac{1}{2}\cos(\phi), -\frac{\sqrt{3}}{2}\cos(\phi)),\quad e_2=(0, \frac{\sqrt{3}}{2}, -\frac{1}{2}).\]
Consider random tangent vectors $U=Z_1e_1+Z_2e_2$, where $Z_1, Z_2$ are two independent uniformly distributed random variables on $[-0.1, 0.1]$. The random outcome $Y$ is obtained as the exponential map at $m_0(X)$ applied to the tangent vector $U$,
\[Y=\mathrm{Exp}_{m_0(X)}(U)=\cos(\|U\|)m_0(X)+\sin(\|U\|)\frac{U}{\|U\|}.\]
A similar generation procedure is used for $T=1$, where the orthonormal basis $(e_1, e_2)$ for the tangent space on $m_1(X)$ is
\[e_1=(\sin(\phi), -\frac{\sqrt{3}}{2}\cos(\phi), -\frac{1}{2}\cos(\phi)),\quad e_2=(0, \frac{1}{2}, -\frac{\sqrt{3}}{2}).\]
Figure~\ref{fig:comp} illustrates randomly generated outcomes using the above generation procedure for a sample size $n=200$, which are seen to be distributed around the true regression functions.

We are interesting in estimating the GATE, with true value $\gamma_{\E_\oplus[Y(0)], \E_\oplus[Y(1)]}$, where 
\[\E_\oplus[Y(0)]=(\frac{\sqrt{2}}{2},\frac{\sqrt{2}}{4}, \frac{\sqrt{6}}{4}),\quad\E_\oplus[Y(1)]=(\frac{\sqrt{2}}{2},\frac{\sqrt{6}}{4}, \frac{\sqrt{2}}{4}).\]
Table~\ref{tab:com} summarizes the simulation results. The ASE of the doubly robust estimator decreases as the sample size increases when either the OR or IPW model is correct, thus again confirming the double robustness property. In comparison, neither the OR nor IPW estimator exhibits the double robustness property: When the corresponding model is misspecified, their ASEs do not converge as the sample size increases.

\begin{figure}[tb]
    \centering
    \includegraphics[width=0.5\linewidth]{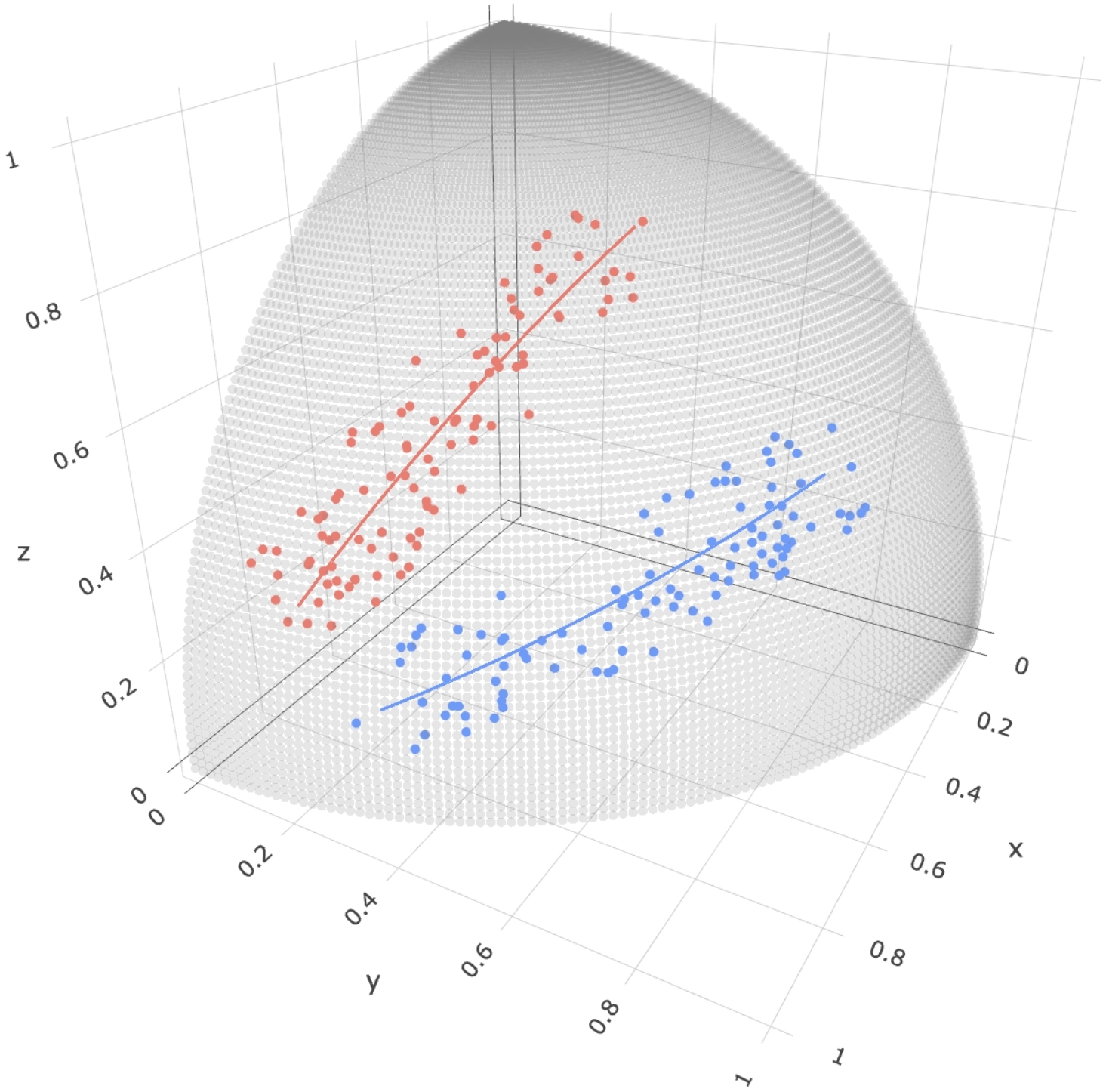}
    \caption{Simulation example for compositional data represented on the sphere with the geodesic metric for a sample of size $n=200$. Red: $T=0$; blue: $T=1$. The true regression functions $m_0(\cdot)$ and $m_1(\cdot)$ are shown as red and blue lines.}
\label{fig:comp}
\end{figure}

\begin{table}[tb]
	\caption{Simulation results for compositional data. Average squared errors and standard deviations (in parentheses) for GATE using four different estimation procedures. The model specifications are as in the first two columns. DR: doubly robust; CF: cross-fitting; OR: outcome regression; IPW: inverse probability weighting. All error values are reported $\times 10^3$ for presentation.}
	\centering
	\begin{tabular}{lllllll}
	\hline
	\multicolumn{2}{c}{Model}&Sample size&\multicolumn{4}{c}{Estimator}
	\\\cline{1-2}\cline{4-7}
	OR&IPW&&DR&CF&OR&IPW
	\\\hline
	\multirow{3}{*}{\ding{51}}&\multirow{3}{*}{\ding{51}}&100&1.369 (1.648)&1.377 (1.635)&1.371 (1.651)&2.719 (2.134)
	\\
    &&300&0.435 (0.545)&0.436 (0.545)&0.436 (0.544)&1.473 (0.843)
    \\
    &&1000&0.126 (0.145)&0.126 (0.145)&0.127 (0.146)&1.065 (0.337)
    \\\hline
    \multirow{3}{*}{\ding{51}}&\multirow{3}{*}{\ding{55}}&100&1.372 (1.651)&1.374 (1.655)&1.371 (1.651)&6.414 (4.402)
	\\
    &&300&0.436 (0.542)&0.437 (0.544)&0.436 (0.544)&5.077 (2.454)
    \\
    &&1000&0.127 (0.145)&0.127 (0.146)&0.127 (0.146)&4.505 (1.262)
    \\\hline
	\multirow{3}{*}{\ding{55}}&\multirow{3}{*}{\ding{51}}&100&3.566 (3.62)&4.171 (4.364)&7.017 (5.047)&2.719 (2.134)
	\\
    &&300&1.808 (1.371)&1.885 (1.456)&5.663 (2.821)&1.473 (0.843)
    \\
    &&1000&1.202 (0.521)&1.210 (0.522)&5.038 (1.438)&1.065 (0.337)
	\\\hline
	\end{tabular}
	\label{tab:com}
\end{table}

\section{Real world applications}\label{sec:data}
\subsection{New York yellow taxi system after COVID-19 outbreak}\label{subsec:taxi}
Yellow taxi trip records in New York City (NYC), containing details such as pick-up and drop-off dates/times, locations, trip distances, payment methods, and driver-reported passenger counts, can be accessed at \url{https://www.nyc.gov/site/tlc/about/tlc-trip-record-data.page}. Additionally, NYC Coronavirus Disease 2019 (COVID-19) data are available at \url{https://github.com/nychealth/coronavirus-data}, providing citywide and borough-specific daily counts of probable and confirmed COVID-19 cases in NYC since February 29, 2020. 

We focused on taxi trip records in Manhattan, which experiences the highest taxi traffic. Following preprocessing procedures outlined in \citet{zhou:22}, we grouped the 66 taxi zones (excluding islands) into 13 regions. We restricted our analysis to the period comprising 172 days from April 12, 2020 to September 30, 2020, during which taxi ridership per day in Manhattan steadily increased following a decline due to the COVID-19 outbreak. For each day, we constructed a daily undirected network with nodes corresponding to the 13 regions and edge weights representing the number of people traveling between connected regions. Self-loops in the networks were removed as the focus was on connections between different regions. Thus, we have observations consisting of a simple undirected weighted network for each of the 172 days, each associated with a $13\times13$ graph Laplacian.

We aim to investigate the causal effect of COVID-19 new cases on daily taxi networks in Manhattan. COVID-19 new cases were dichotomized into 0 if less than 60 and 1 otherwise, resulting in 79 and 93 days classified into low (0) and high (1) COVID-19 new cases groups, respectively. The outcomes are graph Laplacians as discussed in Example \ref{exm:net}. Confounders of interest include a weekend indicator and daily average temperature.

We obtained DR and CF estimators using the proposed approach, along with OR and IPW estimators for comparison. In Figure~\ref{fig:taxiate}, we illustrate the entry-wise differences between the adjacency matrices corresponding to low and high COVID-19 new cases for different estimators. The DR, CF, and OR estimators exhibit similar performance, indicating that high COVID-19 led to less traffic. Regions with the largest differences include regions 105, 106, and 108, primarily residential areas including popular locations such as Penn Station, Grand Central Terminal, and the Metropolitan Museum of Art. The impact of COVID-19 new cases on traffic networks is seen to be negative, especially concerning traffic in residential areas. Conversely, the IPW estimator appears ineffective in capturing this causal effect.

\begin{figure}[tb]
	\centering
	\begin{subfigure}{.4\textwidth}
		\centering
		\includegraphics[width=\linewidth]{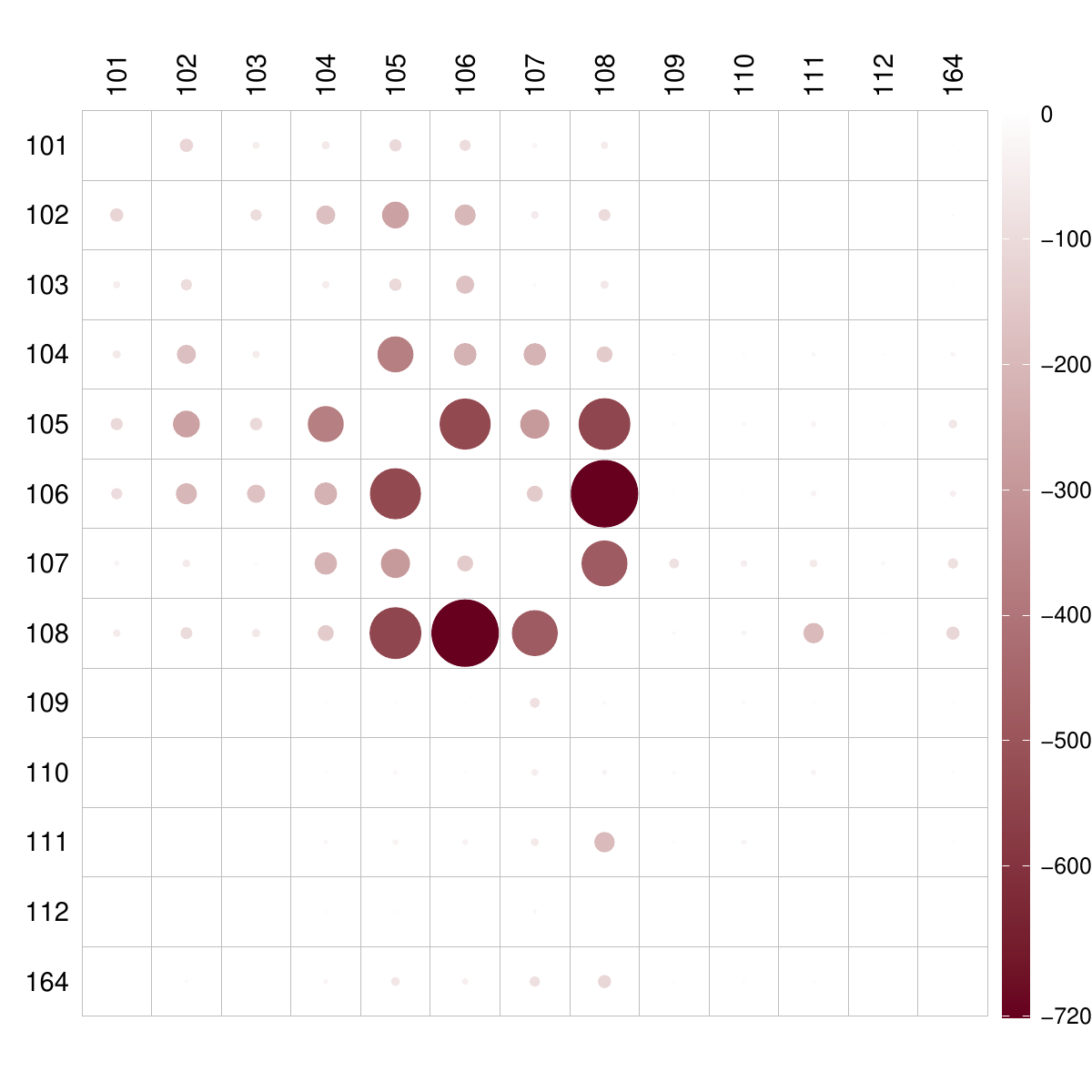}
            \caption{Cross-fitting}
	\end{subfigure}%
	\begin{subfigure}{.4\textwidth}
		\centering
		\includegraphics[width=\linewidth]{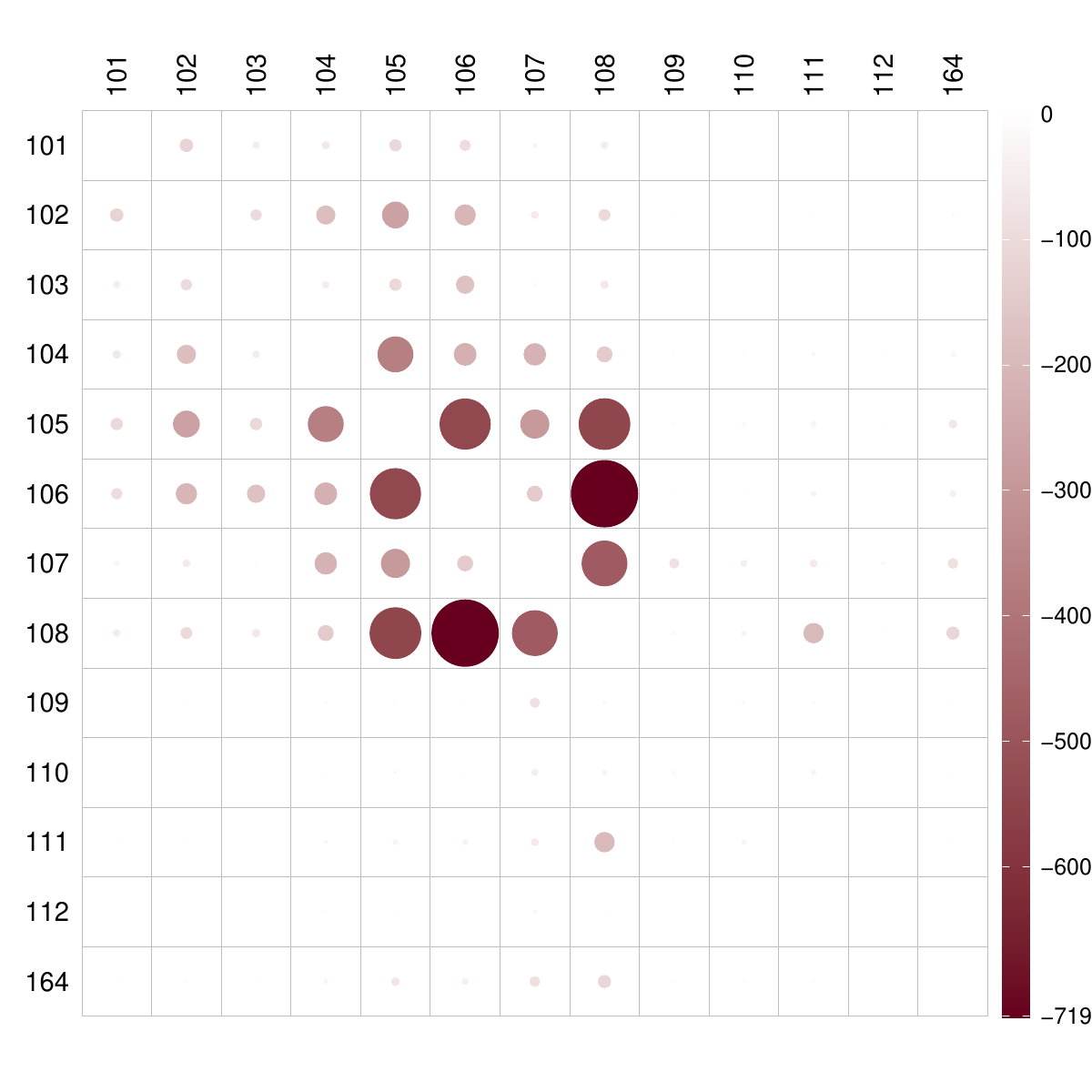}
            \caption{Doubly robust}
	\end{subfigure}
        \begin{subfigure}{.4\textwidth}
		\centering
		\includegraphics[width=\linewidth]{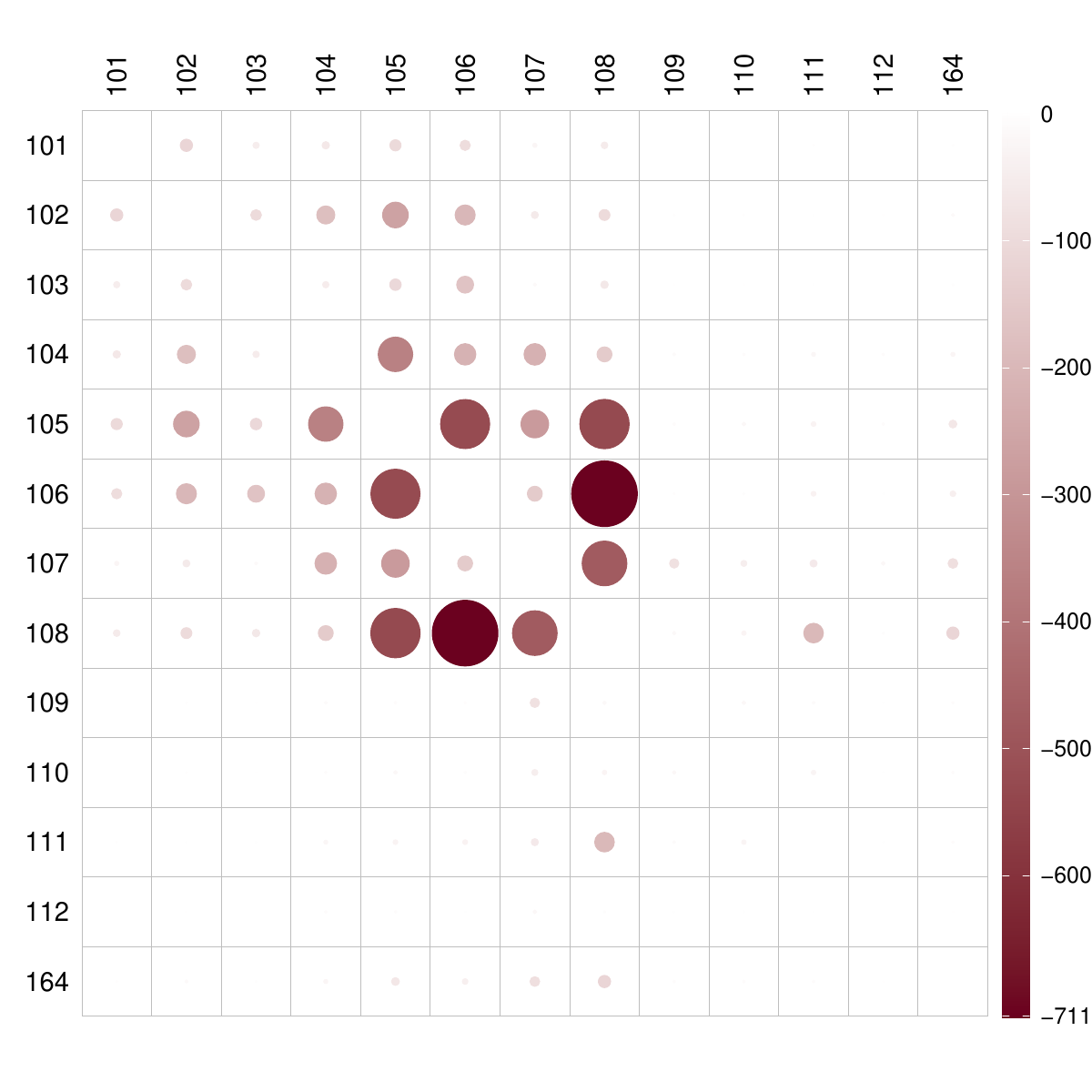}
            \caption{Outcome regression}
	\end{subfigure}%
	\begin{subfigure}{.4\textwidth}
		\centering
		\includegraphics[width=\linewidth]{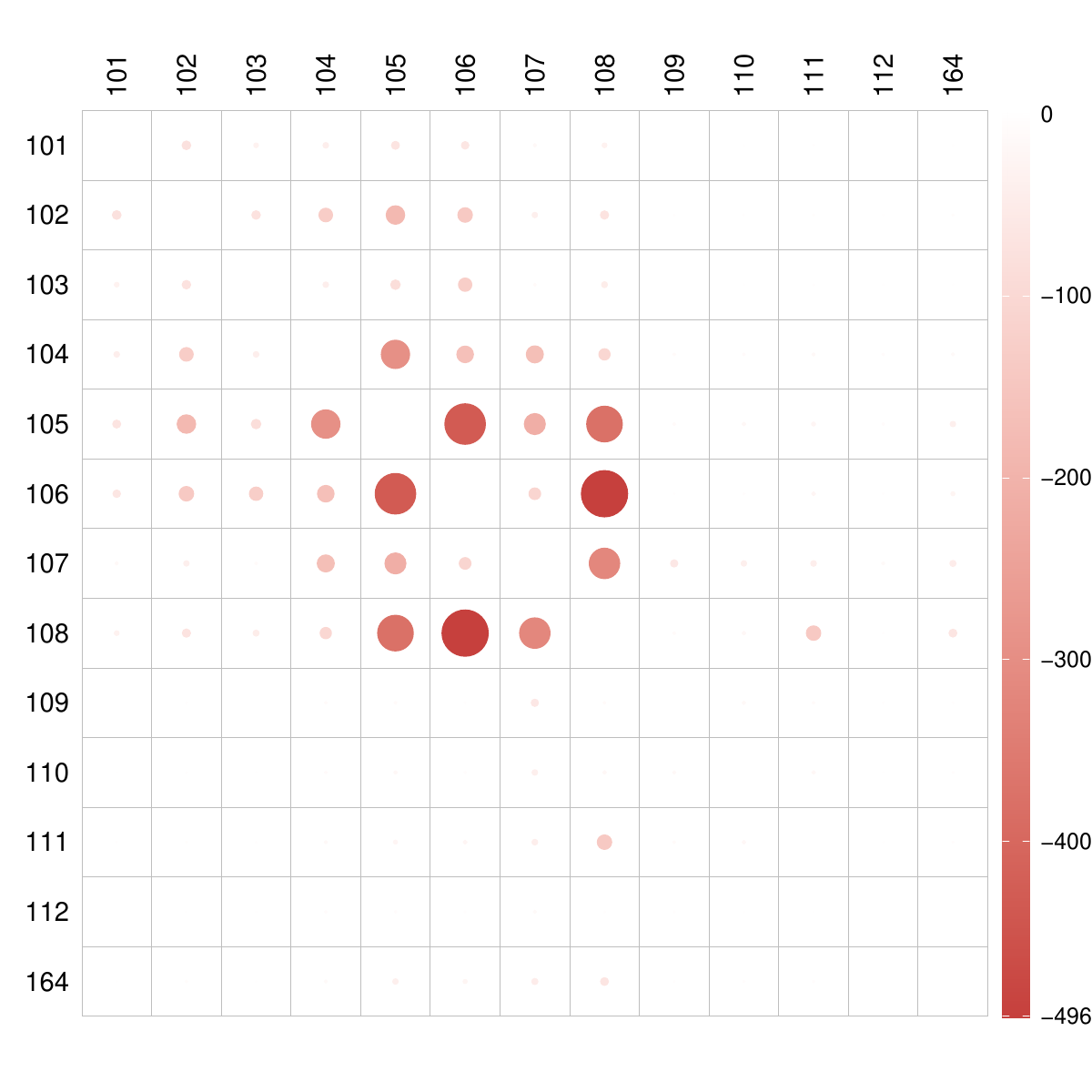}
            \caption{Inverse probability weighting}
	\end{subfigure}
    \caption{Entry-wise differences between the adjacency matrices corresponding to the mean potential networks under low and high COVID-19 new cases, estimated using cross-fitting (CF), doubly robust (DR), outcome regression (OR), and inverse probability weighting (IPW). Each heatmap summarizes the effect of elevated case counts on daily taxi flows among 13 aggregated Manhattan regions. The DR, CF, and OR estimators produce similar spatial patterns, indicating reduced mobility in key residential and commercial areas (regions 105, 106, and 108), while the IPW estimator shows divergent and less stable behavior, likely due to propensity score misspecification.}
\label{fig:taxiate}
\end{figure}

To visualize the GATE in the network setting, Figure~\ref{fig:taxigds} displays the intrinsic geodesic connecting the DR-estimated mean potential networks for low and high COVID-19 new cases. Each panel shows an intermediate network along the geodesic at evenly spaced values of $t \in \{0, 0.2, 0.4, 0.6, 0.8, 1\}$, where $t=0$ corresponds to the mean network under low case counts and $t=1$ corresponds to the mean network under high case counts. The geodesic provides an interpretable description of the causal effect: it depicts the smooth, geometry-respecting deformation of the daily taxi network as COVID-19 activity increases. Specifically, edges associated with travel among residential and commercial hubs (notably regions 105, 106, and 108) weaken progressively along the path, reflecting a coordinated decline in mobility across these areas.

A common alternative is to first reduce each network to a collection of summary statistics, such as degree centrality, clustering coefficients, or average path length, and then apply standard causal inference tools to these Euclidean summaries. While convenient, such an approach captures only selected features of the networks and cannot reveal how connectivity patterns change globally. In contrast, the geodesic in Figure~\ref{fig:taxigds} shows the full structural transition between the mean potential networks, capturing coordinated edge-level changes that are not recoverable from low-dimensional summaries. The GATE therefore provides a coherent representation of how the overall network responds to elevated COVID-19 case counts.

\begin{figure}[tb]
    \centering
	\includegraphics[width=0.8\linewidth]{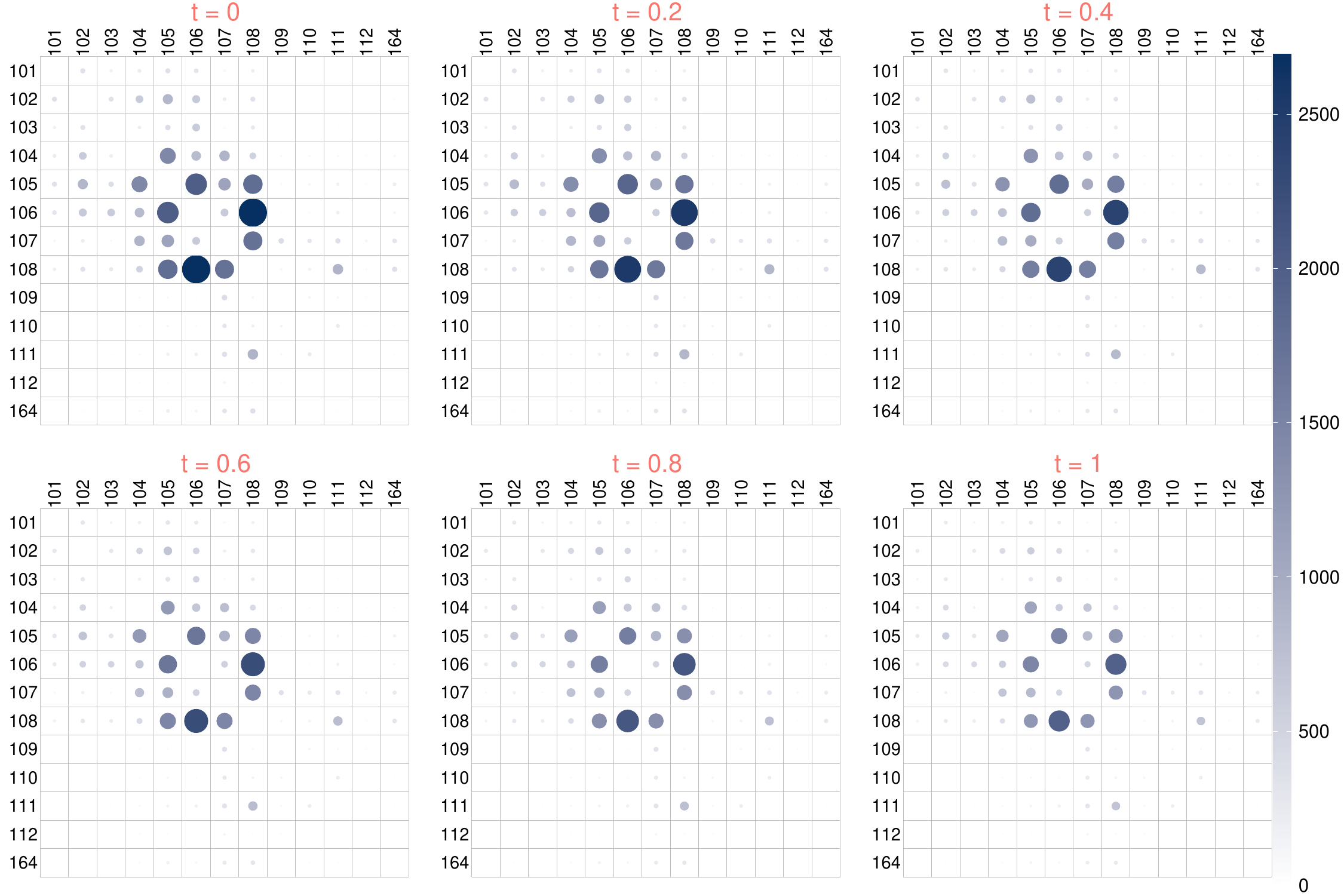}
    \caption{GATE estimated via the doubly robust approach for the NYC taxi network example. Shown are six intermediate networks along the geodesic between the DR-estimated mean potential networks under low ($t=0$) and high ($t=1$) COVID-19 case counts. The geodesic depicts the transition between the two average networks and highlights the systematic weakening of connectivity among key residential and commercial regions as case counts rise.}
    \label{fig:taxigds}
\end{figure}

The magnitude of the estimated GATE under the doubly robust approach, $d(\hat{\Theta}_0^{\mathrm{(DR)}}, \hat{\Theta}_1^{\mathrm{(DR)}})$, is 5216. To evaluate its statistical significance, we apply the inference procedure described in Section~\ref{sec:uq}. The resulting $p$-value is 0.001, providing strong evidence against the null hypothesis of no treatment effect. This suggests that COVID-19 case counts have a statistically significant impact on the daily taxi network structure in Manhattan at the 0.05 significance level.

\subsection{U.S. electricity generation data}\label{subsec:electricity}
Compositional data are ubiquitous and are not situated in a vector space as they correspond to vectors with non-negative elements that sum to 1; see Example~\ref{exm:com}. Various approaches have been developed to address the inherent nonlinearity of compositional data \citep{aitc:86, scea:14, filz:18} which are common in the analysis of geochemical and microbiome data.

Here we focus on U.S. electricity generation data, publicly available on the U.S. Energy Information Administration website (\url{http://www.eia.gov/electricity}). The data reflect net electricity generation from various sources for each state in the year 2020. In preprocessing, we excluded the ``pumped storage'' and ``other'' categories due to data errors and consolidated the remaining energy sources into three categories: Natural Gas (corresponding to the category of ``natural gas''), Other Fossil (combining ``coal,'' ``petroleum,'' and ``other gases''), and Renewables and Nuclear (combining ``hydroelectric conventional,'' ``solar thermal and photovoltaic,'' ``geothermal,'' ``wind,'' ``wood and wood-derived fuels,'' ``other biomass,'' and ``nuclear''). This yielded a sample of $n=50$ compositional observations taking values in the 2-simplex $\Delta^2$, as illustrated with a ternary plot in Figure~\ref{fig:energy}. Following the approach described in Example \ref{exm:com}, we applied the component-wise square root transformation, resulting in compositional outcomes as elements of the sphere $\mathcal{S}^2$, equipped with the geodesic metric.

In our analysis, the exposure (treatment) of interest is whether the state produced coal in 2020, where 29 states produced coal in 2020 while 21 states did not. The outcomes are the compositional data discussed in Example \ref{exm:com}, where we consider two possible confounders: Gross domestic product (GDP) per capita (millions of chained 2012 dollars) and the proportion of electricity generated from coal and petroleum in 2010 for each state. We implemented the proposed approach to obtain the DR, CF, OR, and IPW estimators, with results demonstrated in Figure~\ref{fig:energy}. The DR, CF, and OR estimators yield similar results, suggesting that coal production leads to a smaller proportion of renewables \& nuclear. In contrast, the IPW estimator yields a slightly different results, possibly due to the violation of the propensity score model, and therefore
should not be used here.

\begin{figure}[tb]
    \centering
	\includegraphics[width=0.9\linewidth]{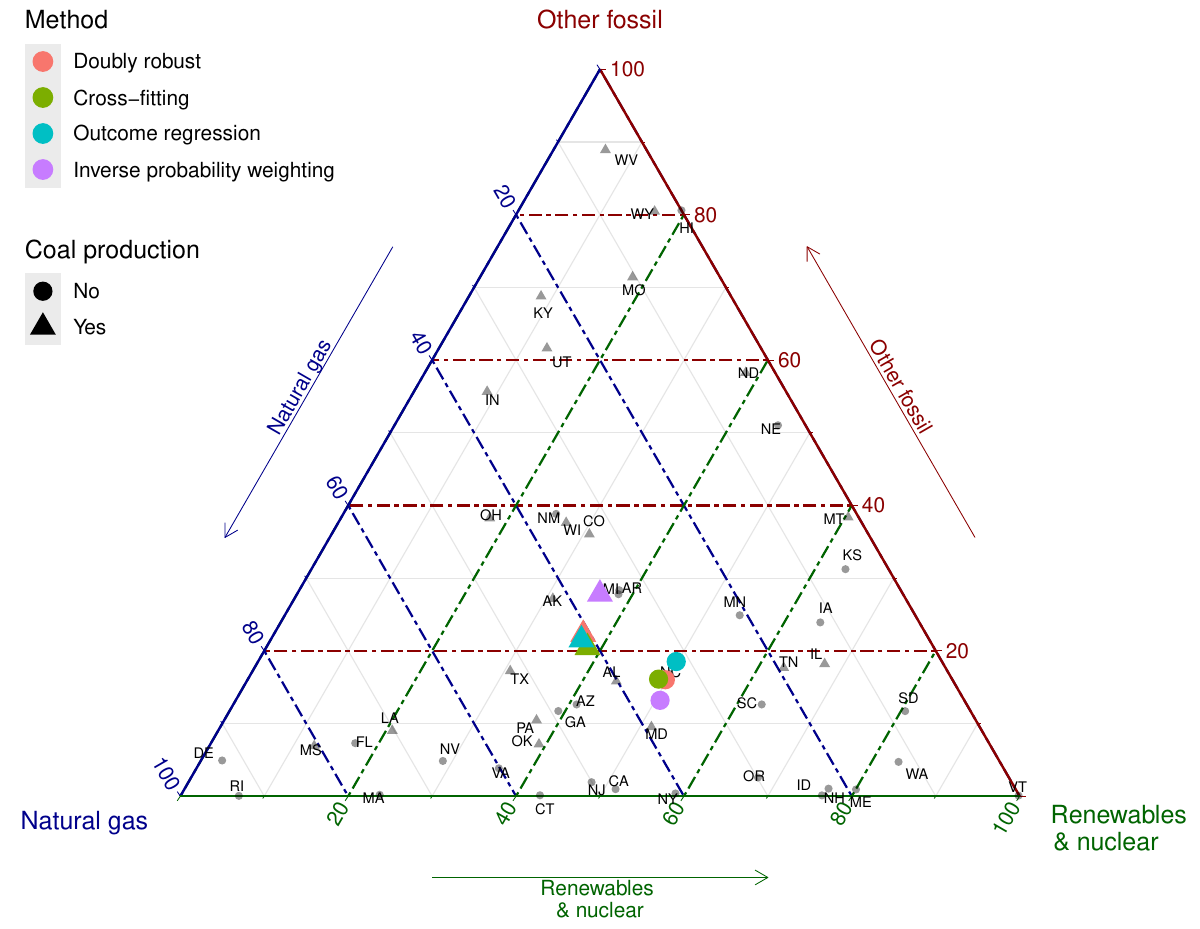}
    \caption{Ternary plot illustrating the composition of electricity generation in 2020 across 50 U.S. states and mean potential outcomes with and without production of coal in 2020 using four different methods. States are shape-coded based on whether they produced coal (triangle) or not (circle) in 2020. Mean potential outcomes are color-coded based on which method is used. The geodesic distance between mean potential outcomes for the DR, CF, OR, and IPW estimators is 0.130, 0.107, 0.131, and 0.198, respectively, where the GATE obtained from IPW differs substantially from the GATE obtained with the other estimators.}
    \label{fig:energy}
\end{figure}

To visualize the GATE obtained under the doubly robust approach, Figure~\ref{fig:energys2} displays the intrinsic geodesic on $\mathcal S^2_{+}$ that connects the DR-estimated mean potential outcomes. This geodesic represents the GATE itself, viewed as the intrinsic shortest path between the two average compositions. Interpreting this path reveals how the electricity mix would evolve when moving from the mean composition of states that do not produce coal to that of states that do. In this example, the geodesic shows a coordinated decline in the Renewables and Nuclear share accompanied by a compensatory increase primarily in Natural Gas. Unlike Euclidean interpolation, which would impose linear and componentwise adjustments, the geodesic reflects a coherent redistribution of energy sources that respects the geometry of compositional data and yields a substantively meaningful depiction of the treatment effect.

\begin{figure}[tb]
    \centering
	\includegraphics[width=0.5\linewidth]{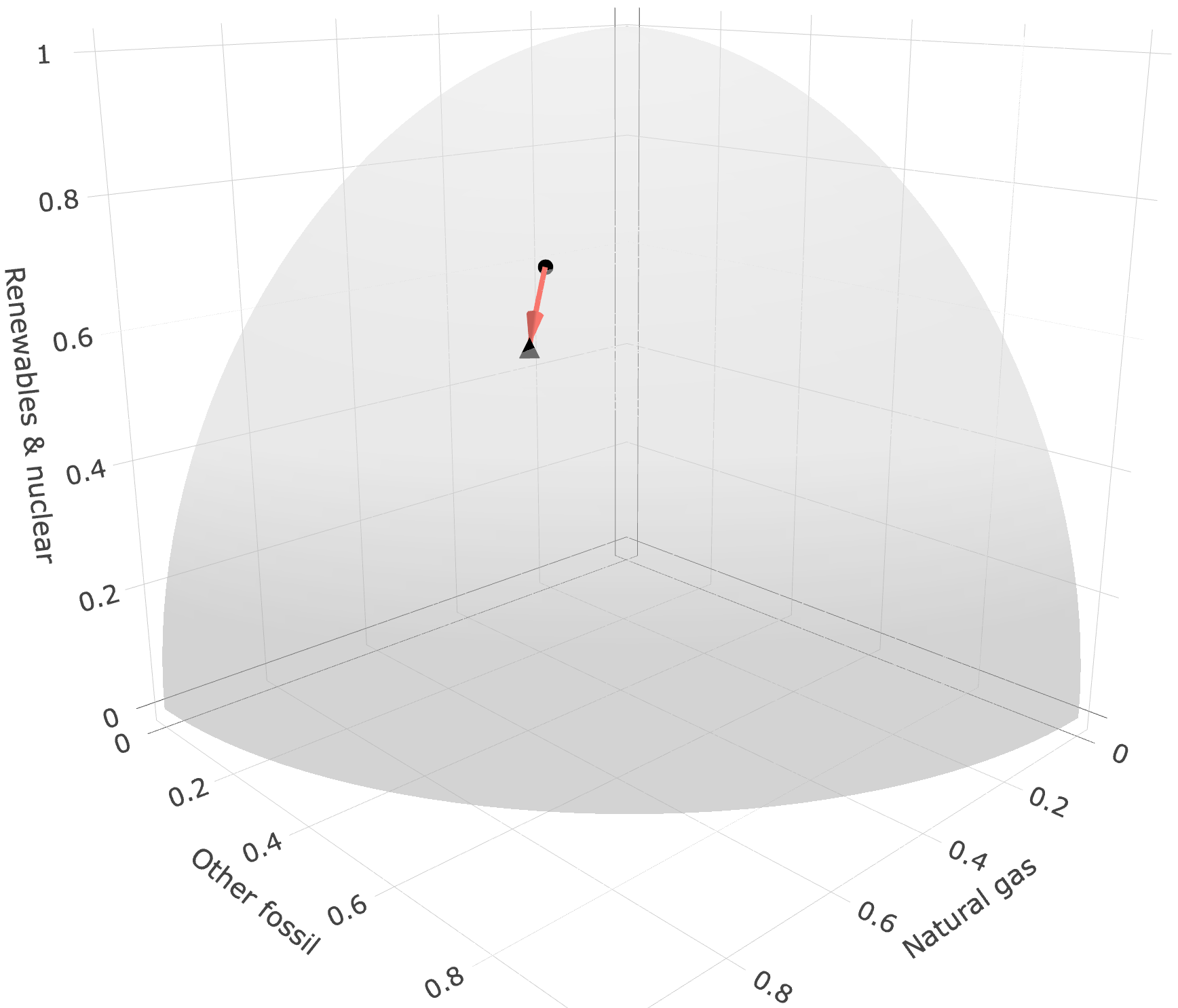}
    \caption{Doubly robust estimate of the GATE visualized on $\mathcal{S}^2_{+}$. The geodesic represents the intrinsic causal effect, illustrating the transition from the mean potential composition under non-coal-producing states (circle) to that under coal-producing states (triangle).}
    \label{fig:energys2}
\end{figure}

An additional application to brain functional connectivity networks is provided in Appendix F.

\section{Discussion} \label{sec:disc} 
We mention here a few additional issues. The assumptions are satisfied for most metric spaces of statistical interest and guarantee relatively fast convergence rates of the proposed doubly robust estimators, but alternative assumptions are also possible. For example, one can replace Assumption \ref{Ass:GDR-rate}(ii) with an analogous condition as Assumption (U2) in \citet{PeMu19}, that is
\[\inf_{d(\nu,\Theta_t^{(\mathrm{DR})})<\eta}\left(Q_t(\nu;e, \mu_t) - Q_t(\Theta_t^{(\mathrm{DR})}; e, \mu_t) - Cd(\nu, \Theta_t^{(\mathrm{DR})})^\beta\right)\geq 0.\]
However, under this relaxed assumption the convergence rates of the estimators $\hat{\Theta}_t^{(\mathrm{DR})}$ and $\hat{\Theta}_t^{(\mathrm{CF})}$ will be slower than those reported in Theorem \ref{Thm:GDR-rate}. Specifically, the term $(\varrho_n + r_n)^{\frac{1-\epsilon}{\beta-1}}$ would need to be replaced with $(\varrho_n + r_n)^{\frac{1}{\beta}}$.

Even when objects that represent outcomes of interest are given, for example SPD matrices, there is still the issue of choosing a metric and ensuing geodesics, where interpretability of the resulting geodesic average treatment effects will be a consideration. For a brief discussion of this and related issues see \citet{mull:24}.

There are numerous scenarios where metric space-valued outcomes are of interest in causal inference, as we have demonstrated in the data application section. The tools to construct estimators for treatment effects are inspired by their linear counterparts but are substantially different, and notions of Euclidean treatment effects need to be revisited to arrive at canonical generalizations for metric spaces. These investigations also shed new light on and lead to a reinterpretation of the classical Euclidean case. Future exploration of technical aspects and real-world applications for object data will enhance the applicability of causal inference across a range of complex data.

\appendix

\section{Asymptotic properties of OR and IPW estimators}\label{apdx:OR-IPW}

\subsection{Outcome regression}\label{apdx:OR}
In this section, we provide asymptotic properties of outcome regression estimators. 

\begin{assumption}\label{Ass:model-OR} 
For $t\in\{0, 1\}$, let $\hat{\mu}_t(\cdot)$ be estimators for the outcome regression functions $m_t(\cdot)$. Assume that $\sup_{x \in \mathcal{X}}d(\hat{\mu}_t(x),m_t(x)) = O_{\Prob}(r_{n,1})$, $t \in \{0,1\}$ with $r_{n,1} \to 0$ as $n \to \infty$. 
\end{assumption}

\begin{assumption}\label{Ass:OR-criterion}
Define $\Theta_t^{(\mathrm{OR})} = \argmin_{\nu \in \mathcal{M}}Q^{(\mathrm{OR})}(\nu;m_t)$ where $Q^{(\mathrm{OR})}(\nu;\mu) = \E[d^2(\nu, \mu(X))]$. The population minimizer $\Theta_t^{(\mathrm{OR})}$ and its empirical counterpart $\hat{\Theta}_t^{(\mathrm{OR})}$ exist and are unique. Moreover, for any $\varepsilon>0$,
\[\inf_{d(\nu, \Theta_t^{(\mathrm{OR})})>\varepsilon}Q^{(\mathrm{OR})}(\nu;m_t) > Q^{(\mathrm{OR})}(\Theta_t^{(\mathrm{OR})};m_t).\]
\end{assumption}

\begin{assumption}\label{Ass:OR-rate}
For $t \in \{0,1\}$, 
\begin{itemize}
\item[(i)] As $\delta \to 0$, 
\begin{align*}
J_t(\delta)&:=\int_0^1 \sqrt{1 + \log N(\delta \varepsilon, B_\delta(\Theta_t^{(\mathrm{OR})}),d)}d\varepsilon = O(1), \\
J_{m_t}(\delta) &:= \int_0^1 \sqrt{1 + \log N(\delta \varepsilon, B_{\delta'_1}(m_t),d_{\infty})}d\varepsilon = O(\delta^{-\vartheta_1}) 
\end{align*}
for some $\delta'_1>0$ and $\vartheta_1 \in (0,1)$.
\item[(ii)] There exist positive constants $\eta$, $\eta_1$, $C$, $C_1$, and $\beta>1$ such that 
\begin{align*}
\inf_{d_\infty(\mu,\mu_t)\leq \eta_1}
\inf_{d(\nu, \Theta_t^{(\mathrm{OR})})\leq \eta} 
\Bigg\{ & Q^{(\mathrm{OR})}(\nu;\mu) - Q^{(\mathrm{OR})}(\Theta_t; \mu) - \\
&Cd(\nu, \Theta_t^{(\mathrm{OR})})^\beta  
+ C_1\eta_1^{\frac{\beta}{2(\beta-1)}}d(\nu, \Theta_t^{(\mathrm{OR})})^{\frac{\beta}{2}} 
\Bigg\} \geq 0.
\end{align*}
\end{itemize}
\end{assumption}

The following results provide the consistency and convergence rates of the outcome regression estimators. The proofs are similar to those of Theorems 4.1 and 4.2 and thus omitted. 

\begin{theorem}\label{Thm:OR-conv}
Suppose that Assumptions 3.1(ii), \ref{Ass:model-OR}, and \ref{Ass:OR-criterion} hold. Then as $n \to \infty$, 
\[d(\hat{\Theta}_t^{(\mathrm{OR})},\E_\oplus[Y(t)]) = o_{\Prob}(1),\quad t \in \{0,1\}.\] 
\end{theorem}

\begin{theorem}\label{Thm:OR-rate}
Suppose that Assumptions 3.1(ii), \ref{Ass:model-OR}, \ref{Ass:OR-criterion}, and \ref{Ass:OR-rate} hold. Then for each $\epsilon \in (0,1)$, as $n \to \infty$, 
\[d(\hat{\Theta}_t^{(\mathrm{OR})},\E_\oplus[Y(t)]) = O_{\Prob}\left( n^{-\frac{1}{2(\beta - 1 +\vartheta_1)}} + r_{n,1}^{\frac{1-\epsilon}{(\beta-1)}}\right),\quad t \in \{0,1\}.\]
\end{theorem}

\subsection{Inverse probability weighting}\label{apdx:IPW}
In this section, we provide asymptotic properties of inverse probability weighting estimators. Recall that for $t\in\{0, 1\}$,
\[\Theta_t^{(\mathrm{IPW})}=\E_{\oplus}\left[\gamma_{y_\oplus, Y}\left(\frac{tT}{e(X)} + \frac{(1-t)(1-T)}{1-e(X)}\right)\right]=\argmin_{\nu\in\mathcal{M}}Q_t^{(\mathrm{IPW})}(\nu;e),\]
where
\[Q_t^{(\mathrm{IPW})}(\nu;e) = \E\left[d^2\left(\nu, \gamma_{y_\oplus,Y}\left(\frac{tT}{e(X)} + \frac{(1-t)(1-T)}{1-e(X)}\right)\right)\right].\]

\begin{assumption}\label{Ass:model-IPW} 
Let $\hat{e}(\cdot)$ denote the estimator of the propensity score $p(\cdot)$. Assume that $\sup_{x \in \mathcal{X}}|\hat{e}(x) -p(x)| = O_{\Prob}(\varrho_{n,1})$ with $\varrho_{n,1} \to 0$ as $n \to \infty$. 
\end{assumption}

\begin{assumption}\label{Ass:IPW-criterion}
The population minimizer $\Theta_t^{(\mathrm{IPW})}$ and its empirical counterpart $\hat{\Theta}_t^{(\mathrm{IPW})}$ exist and are unique. Moreover, for any $\varepsilon>0$,
\[\inf_{d(\nu, \Theta_t^{(\mathrm{IPW})})>\varepsilon}Q_t^{(\mathrm{IPW})}(\nu;e) > Q_t^{(\mathrm{IPW})}(\Theta_t^{(\mathrm{IPW})};e).\]
\end{assumption}

\begin{assumption}\label{Ass:IPW-rate}
For $t \in \{0,1\}$, 
\begin{itemize}
\item[(i)] As $\delta \to 0$, 
\begin{align*}
J_t(\delta)&:=\int_0^1 \sqrt{1 + \log N(\delta \varepsilon, B_\delta(\Theta_t^{(\mathrm{IPW})}),d)}d\varepsilon = O(1),\\
J_p(\delta) &:= \int_0^1 \sqrt{1 + \log N(\delta \varepsilon, B_{\delta'_1}(p),\|\cdot\|_{\infty})}d\varepsilon = O(\delta^{-\vartheta_2})
\end{align*}
for some $\delta'_1>0$ and $\vartheta_2 \in (0,1)$, where for $p_1,p_2:\mathcal{X} \to (0, 1)$, $\|p_1 - p_2\|_\infty=\sup_{x \in \mathcal{X}}|p_1(x) - p_2(x)|$.

\item[(ii)] There exist positive constants $\eta$, $\eta_1$, $C$, $C_1$, and $\beta>1$ such that  
\begin{align*}
&\inf_{\|e-p\|_\infty \leq \eta_1}\hspace{0.2em}\inf_{d(\nu,\Theta_t^{(\mathrm{IPW})})\leq \eta}\!\!\left\{\!Q_t^{(\mathrm{IPW})}\!(\nu;e) \!-\! Q_t^{(\mathrm{IPW})}\!(\Theta_t; e) \right. \\ 
&\left. \qquad \qquad \qquad \qquad \qquad - d(\nu, \Theta_t^{(\mathrm{IPW})})^\beta \!+\! C_1\eta_1^{\frac{\beta}{2(\beta-1)}}\!d(\nu, \Theta_t^{(\mathrm{IPW})})^{\frac{\beta}{2}}\!\right\} \geq 0. 
\end{align*} 
\end{itemize}
\end{assumption}

The following results provide the consistency and convergence rates of the inverse probability weighting estimators. The proofs follow by the same empirical process arguments as in Theorems 4.1 and 4.2.

\begin{theorem}\label{Thm:IPW-conv}
Suppose that Assumptions 3.1, 4.2, \ref{Ass:model-IPW}, and \ref{Ass:IPW-criterion} hold. Then as $n \to \infty$, 
\[d(\hat{\Theta}_t^{(\mathrm{IPW})},\E_\oplus[Y(t)]) = o_{\Prob}(1),\quad t \in \{0,1\}.\]  
\end{theorem}

\begin{theorem}\label{Thm:IPW-rate}
Suppose that Assumptions 3.1, 4.2, \ref{Ass:model-IPW}, \ref{Ass:IPW-criterion}, and \ref{Ass:IPW-rate} hold. Then for each $\epsilon \in (0,1)$, as $n \to \infty$, 
\[d(\hat{\Theta}_t^{(\mathrm{IPW})},\E_\oplus[Y(t)]) = O_{\Prob}\left( n^{-\frac{1}{2(\beta - 1 +\vartheta_2)}} + \varrho_{n,1}^{\frac{1-\epsilon}{(\beta-1)}}\right),\quad t \in \{0,1\}.\]
\end{theorem}

\section{Verification of model assumptions}\label{apdx:ass}

In this section, we verify our assumptions for Examples 1--5. 

\begin{proposition}
    \label{prop:net}
    The space of networks defined in Example 1 satisfies Assumptions 3.2, 4.2, 4.3(i), and 4.4.
\end{proposition}
\begin{proof}
    We first verify Assumptions~4.2, 4.3(i), and 4.4. In the space of graph Laplacians equipped with the Frobenius metric, geodesics coincide with Euclidean line segments connecting the endpoints. It follows directly that Assumption~4.2 holds with constant $C_0 = 1/\eta_0 - 1$. Moreover, by adapting the argument in Theorem~2 of \citet{zhou:22}, one can show that Assumptions~4.3(i) and 4.4 are satisfied with $\beta = 2$ in Assumption~4.4(ii).

    We next verify Assumption~3.2. For any random graph Laplacian $Y$, consider the model $Y=m(X)+\epsilon$ where $\varepsilon$ is an $m \times m$ random graph Laplacian satisfying $\E[\varepsilon_{jk}|X] = 0$ and $\E[\varepsilon_{jk}^2|X] < \infty$ for $1 \leq j < k \leq m$. Let $\mathrm{tr}(A)$ denote the trace of a matrix $A$. For any graph Laplacian $\nu$, the squared Frobenius distance satisfies
    \begin{align*}
     E[d_F^2(\nu,Y)|X] &= \E[d_F^2(\nu,m(X))|X] + \E\left[\mathrm{tr}\left(\varepsilon'\varepsilon\right)|X\right],\\
    \E[d_F^2(\nu,Y)] &= \E[d_F^2(\nu,m(X))] + \E\left[\mathrm{tr}\left(\varepsilon'\varepsilon\right)\right],
    \end{align*} 
    This decomposition implies that $\Eo[Y|X]=m(X)$ and $\Eo[Y] = \Eo[m(X)]$, thereby verifying Assumption~3.2(i). Assumption~3.2(ii) follows directly from the linear structure of the Frobenius metric, under which geodesic interpolation reduces to affine combinations.
\end{proof}

\begin{proposition}
    \label{prop:cor}
    The space of SPD matrices defined in Example 2 satisfies Assumptions 3.2, 4.2, 4.3(i), and 4.4.
\end{proposition}
The proof of Proposition \ref{prop:cor} is similar to that of Proposition \ref{prop:net} and is omitted.

\begin{proposition}
    \label{prop:fun}
    The Hilbert space $L^{2}(\mathcal{T})$ defined in Example 3 satisfies Assumptions 3.2, 4.2, 4.3(i), and 4.4.
\end{proposition}
\begin{proof}
In the Hilbert space $L^2(\mathcal{T})$, geodesics are given by linear interpolation, $\gamma_{f,g}(t)=f+t(g-f)$. Hence, for any $f_1,f_2,g\in L^2(\mathcal{T})$ and 
$\kappa\in[1/(1-\eta_0),1/\eta_0]$,
\[\begin{aligned}
d(\gamma_{f_1,g}(\kappa),\gamma_{f_2,g}(\kappa))
&=\| (1-\kappa)f_1+\kappa g
      -(1-\kappa)f_2-\kappa g\|_{L^2}  \\
&=|1-\kappa|\,\|f_1-f_2\|_{L^2}  \\
&\leq\left(\frac{1}{\eta_0}-1\right)d(f_1,f_2),
\end{aligned}\]
so Assumption~4.2 holds with $C_0=1/\eta_0-1$.

Assumption~4.3(i) follows from the strict convexity of the squared Hilbert norm. Indeed, for any random element $Y\in L^2(\mathcal{T})$ with finite second moment,
\[\nu\mapsto \E\|Y-\nu\|_{L^2}^2\]
is uniquely minimized at $\E[Y]$, and the separation condition follows from the quadratic identity
\[\E[\|Y-\nu\|_{L^2}^2]
=
\E[\|Y-\E[Y]\|_{L^2}^2]+\|\nu-\E[Y]\|_{L^2}^2.\]
Assumption~4.4 follows from the standard entropy bounds for bounded subsets of finite-dimensional or suitably regular functional subspaces of $L^2(\mathcal{T})$, together with the same quadratic curvature condition, with $\beta=2$ in Assumption~4.4(ii).

It remains to verify Assumption~3.2. Since $L^2(\mathcal{T})$ is linear, the \f mean coincides with the usual Bochner expectation. Therefore, for any random element $Y\in L^2(\mathcal{T})$ and random vector $X$,
\[\E_\oplus[Y]=\E[Y]=\E[\E[Y|X]]=\E_\oplus[\E_\oplus[Y|X]],\]
which proves Assumption~3.2(i).

For Assumption~3.2(ii), using the linear form of the geodesic and the independence of $U$ and $Y$,
\[\begin{aligned}
\E_\oplus[\gamma_{f,Y}(U)]
&=
\E[f+U(Y-f)]  \\
&=
f+\E[U](\E[Y]-f)  \\
&=
\gamma_{f,\E_\oplus[Y]}(\E[U]).
\end{aligned}\]
Thus Assumption~3.2(ii) holds, completing the proof.
\end{proof}

\begin{proposition}
    \label{prop:mea}
    The Wasserstein space defined in Example 4 satisfies Assumptions 3.2, 4.2, 4.3(i), and 4.4.
\end{proposition}

\begin{proof}
We first verify Assumptions~4.2, 4.3(i), and 4.4. For any probability measure $\mu \in \mathcal{W}$, let $F_\mu^{-1}$ denote its quantile function. The map $Q:\mu \mapsto F_\mu^{-1}$ is an isometry from $\mathcal{W}$ into the Hilbert space $L^2(0,1)$, where the image consists of equivalence classes of left-continuous nondecreasing functions on $(0,1)$ \citep{bigo:17}. Under this representation, the Wasserstein distance $d_{\mathcal{W}}$ coincides with the $L^2$ distance between quantile functions. Consequently, $\mathcal{W}$ can be viewed as a closed and convex subset of $L^2(0,1)$.

In the ambient Hilbert space $(L^2(0,1), d_{L^2})$, geodesics are given by linear interpolation, and it follows directly that Assumption~4.2 holds with constant $C_0 = 1/\eta_0 - 1$. Since $\mathcal{W}$ is a closed subset of $L^2(0,1)$, geodesics in $\mathcal{W}$ are obtained by restricting these interpolations to remain within the admissible set; this restriction can only decrease distances, so Assumption~4.2 continues to hold with the same constant. Furthermore, by adapting the argument in Proposition~1 of \citet{PeMu19}, Assumptions~4.3(i) and 4.4 are satisfied with $\beta = 2$ in Assumption~4.4(ii).

We next verify Assumption~3.2(i). Under the quantile representation, the \f mean in $\mathcal{W}$ corresponds to the pointwise $L^2$ mean of quantile functions. For any random element $Y \in \mathcal{W}$, consider the model
\[F_Y^{-1}(s) = F_{m(X)}^{-1}(s) + \varepsilon, \quad s \in [0,1],\]
where $m(X) \in \mathcal{W}$ and $\varepsilon$ is a real-valued random variable satisfying $\E[\varepsilon|X] = 0$ and $\E[\varepsilon^2|X] < \infty$. Then, for any $\nu \in \mathcal{W}$,
\begin{align*}
\E\big[d_{\mathcal{W}}^2(\nu, Y)|X\big]&=\E\big[d_{\mathcal{W}}^2(\nu, m(X))|
X\big]+\E[\varepsilon^2|X],\\
\E\big[d_{\mathcal{W}}^2(\nu, Y)\big]&=\E\big[d_{\mathcal{W}}^2(\nu, m(X))\big]+\E[\varepsilon^2],
\end{align*}
which implies $\Eo[Y|X]=m(X)$ and $\Eo[Y] = \Eo[m(X)]$, thereby establishing Assumption~3.2(i).

For Assumption~3.2(ii), note that geodesics in $\mathcal{W}$ correspond to linear interpolation of quantile functions, that is,
\[F_{\gamma_{\alpha,Y}(u)}^{-1} = F_\alpha^{-1} + u\big(F_Y^{-1} - F_\alpha^{-1}\big).\]
Taking expectations and using linearity of the $L^2$ structure, we obtain
\[\E\big[F_{\gamma_{\alpha,Y}(U)}^{-1}\big] 
= F_\alpha^{-1} + \E[U]\big(\E[F_Y^{-1}] - F_\alpha^{-1}\big)
= F_{\gamma_{\alpha,\,\E_\oplus[Y]}(\E[U])}^{-1},\]
which implies $\E_\oplus[\gamma_{\alpha,Y}(U)] = \gamma_{\alpha,\E_\oplus[Y]}(\E[U])$. This establishes Assumption~3.2(ii).
\end{proof}

\begin{proposition}
    \label{prop:com}
    Consider the space of compositional data $(\mathcal{S}^{d-1}_+, d_g)$ defined in Example 5. Let $T_\mu \mathcal{M}$, $\mathrm{Exp}_\mu$ and $\mathrm{Log}_\mu$ denote the tangent bundle, exponential and logarithmic maps at $\mu \in \mathcal{M}$, respectively. Define $\tilde{\nu} = \mathrm{Log}_{\Theta_t^{(\mathrm{DR})}}(\nu)$ and $\tilde{\gamma} = \mathrm{Log}_{\gamma_t}\left(\bar{\gamma}_t\right)$,  where 
    \begin{align*}
    \gamma_t &=\gamma_{\mu_t(X),Y(t)}\left(\frac{tT}{e(x)} + \frac{(1-t)(1-T)}{1-e(X)}\right),\\\bar{\gamma}_t &= \gamma_{\tilde{\mu}(X),Y(t)}\left(\frac{tT}{\tilde{e}(x)} + \frac{(1-t)(1-T)}{1-\tilde{e}(X)}\right).
    \end{align*}
    Additionally, for $v \in T_{\Theta_t^{(\mathrm{DR})}}\mathcal{M}$ and $w \in T_{\gamma_t}\mathcal{M}$, define $g(v,w) = d_g^2(\mathrm{Exp}_{\Theta_t^{(\mathrm{DR})}}(v),\mathrm{Exp}_{\gamma_t}(w))$. If there exist positive constants $\bar{\eta}_1$ and $\bar{\eta}$ such that
    \begin{align}
    &\inf_{\substack{\|v^*\|\leq \bar{\eta}_1\\ \sup_{x \in \mathcal{X}}d_g(\tilde{\mu}(x),\mu_t(x))\leq \bar{\eta}\\ \sup_{x\in\mathcal{X}}|\tilde{e}(x)-e(x)|\leq \bar{\eta}}}\lambda_{\mathrm{min}}\left(\left.\frac{\partial^2}{\partial v \partial v'}\E[g(v,\tilde{\gamma})]\right|_{v = v^*} \right)> 0, \label{ineq:com-eigen}
    \end{align}
    then the space $(\mathcal{S}^{d-1}_+, d_g)$ satisfies Assumptions 3.2(i), 4.2, 4.3(i), and 4.4, where $\lambda_{\mathrm{min}}(A)$ is the smallest eigenvalue of a square matrix $A$.
\end{proposition}
\begin{proof}
    Assumption 4.2 holds trivially for $\alpha_1=\alpha_2$. So we verify Assumption 4.2 for $\alpha_1 \neq \alpha_2$. Note that for any $q_1,q_2 \in \mathcal{S}_+^{d-1}$, we have $\|q_1 - q_2\| \leq d_g(q_1,q_2) \leq \frac{\pi}{2\sqrt{2}}\|q_1- q_2\|$, which implies the equivalence of the geodesic metric and the ambient Euclidean metric. One can find a positive constant $C'$ depending only on $\eta_0$ such that 
    \[\sup_{\beta \in \mathcal{S}^{d-1}_+, \kappa \in [1/(1-\eta_0),1/\eta_0]}\|\gamma_{\alpha_1,\beta}(\kappa) - \gamma_{\alpha_2,\beta}(\kappa)\| \leq C'\|\alpha_1 - \alpha_2\|\] 
    and this yields (4.1) with $C_0 = C'\pi/(2\sqrt{2})$. 
    
    Following the argument in Proposition 3 of \citet{PeMu19}, one can verify that Assumptions 4.3(i) and 4.4 hold with $\beta=2$ in Assumption 4.4(ii) under (\ref{ineq:com-eigen}).
    
    Let $m(x)$ be a function that takes values in $\mathcal{S}_+^{d-1}$. For $y_\oplus=\Eo[m(X)]\in\mathcal{S}^{d-1}_+$, let $T_{y_\oplus} \mathcal{S}_+^{d-1}$ denote the tangent space of $\mathcal{S}_+^{d-1}$ at $y_\oplus$. Consider a random variable $U \in T_{y_\oplus}\mathcal{S}_+^{d-1}$ independent of $X$ such that $\E[U]=0$. Define $Y= \mathrm{Exp}_{m(X)}(\mathcal{P}_{y_\oplus \to m(X)}(U))$ where $\mathcal{P}_{y_\oplus \to m(X)}$ is the parallel transport map from $y_\oplus$ to $m(X)$. Observe that $\E[\mathrm{Log}_{m(X)}(Y)|X]=\E[\mathcal{P}_{y_\oplus \to m(X)}(U)|X] = \mathcal{P}_{y_\oplus \to m(X)}(\E[U|X])=0$. This and the first-order optimality condition of the \f mean on the positive orthant of the unit sphere yield $\Eo[Y|X]=m(X)$. Assume that $U \stackrel{d}{=}RU$ and $m(X) \stackrel{d}{=} Rm(X)$ for all $R \in G_{y_\oplus}=\{R \in \mathrm{SO}(d): Ry_\oplus=y_\oplus\}$ where $\mathrm{SO}(d)$ is the special orthgonal group. Note that for any $R\in G_{y_\oplus}$, 
    \[
    \mathrm{Exp}_p(Rv) = R\mathrm{Exp}_{y_\oplus}(v),\ \mathrm{Log}_{y_\oplus}(Rp)=R\mathrm{Log}_{y_\oplus}(p),\ \mathcal{P}_{Ry_\oplus \to Rp}(Rv) = R\mathcal{P}_{y_\oplus \to p}(v). 
    \]
    Combining the conditions on $U$ and $m(X)$, we have 
    \begin{align*}
    RY 
    &= \mathrm{Exp}_{m(X)}(R\mathcal{P}_{y_\oplus \to m(X)}(U))
    = \mathrm{Exp}_{m(X)}(\mathcal{P}_{Ry_\oplus \to Rm(X)}(RU))\\
    &= \mathrm{Exp}_{m(X)}(\mathcal{P}_{y_\oplus \to Rm(X)}(RU))
    \stackrel{d}{=} \mathrm{Exp}_{m(X)}(\mathcal{P}_{y_\oplus \to m(X)}(U))
    = Y
    \end{align*}
    for all $R \in G_{y_\oplus}$. Then $R\E[\mathrm{Log}_{y_\oplus}(Y)]=\E[\mathrm{Log}_{y_\oplus}(Y)]$. This yields $\E[\mathrm{Log}_{y_\oplus}(Y)]=0$ and from the first-order optimality condition of the \f mean, we have $\E_\oplus[Y]=y_\oplus=\E_\oplus[m(X)]$. Therefore, we can verify Assumption 3.2(i).
\end{proof}

\section{Proofs for Section 3}\label{apdx:main}
\begin{proof}[Proof of Proposition 3.1]
It suffices to show that, for each $t \in \{0,1\}$,
\[\Theta_t = \E_{\oplus}[Y(t)]\]
whenever either $e(\cdot) = p(\cdot)$ or $\mu_t(\cdot) = m_t(\cdot)$. We provide the argument for $t=1$, as the case $t=0$ follows analogously.

For $t=1$, note that the definition of $\Theta_t$ simplifies to
\[\Theta_1 = \E_{\oplus}\left[\gamma_{\mu_1(X),Y(1)}\left(\frac{T}{e(X)}\right)\right].\]
since $Y = Y(1)$ when $T=1$. Define the pseudo-outcome
\[Z_1=\gamma_{\mu_1(X),Y(1)}\left(\frac{T}{e(X)}\right).\]

We first analyze the conditional \f mean of $Z_1$ given $X$. By Assumption~3.2(ii), which characterizes the behavior of \f means under geodesic interpolation, together with Assumption~3.1(ii) (unconfoundedness), which ensures that $Y(1)$ and $T$ are conditionally independent given $X$, we obtain
\begin{align*}
\E_{\oplus}[Z_1|X]
&= \E_{\oplus}\left[\gamma_{\mu_1(X),Y(1)}\left(\frac{T}{e(X)}\right)\middle| X\right] \\
&= \gamma_{\mu_1(X),\E_{\oplus}[Y(1)|X]}\left(\E\!\left[\frac{T}{e(X)} \middle| X\right]\right).
\end{align*}

We now evaluate each component in the above expression. By the definition of the conditional \f mean, we have $\E_{\oplus}[Y(1)|X] = m_1(X)$. Furthermore, by the definition of the propensity score, it follows that
\[\E\left[\frac{T}{e(X)} \middle| X\right] = \frac{p(X)}{e(X)}.\]
Combining these, we obtain
\[\E_{\oplus}[Z_1|X] = \gamma_{\mu_1(X),m_1(X)}\left(\frac{p(X)}{e(X)}\right).\]

Next, we take the marginal \f mean. By Assumption~3.2(i), which serves as the analogue of the law of iterated expectations, we have
\[\Theta_1 
= \E_{\oplus}[Z_1] 
= \E_{\oplus}\big[\E_{\oplus}[Z_1|X]\big]
= \E_{\oplus}\left[\gamma_{\mu_1(X),\, m_1(X)}\left(\frac{p(X)}{e(X)}\right)\right].\]

We now consider the two cases corresponding to double robustness.

\medskip
\noindent\textbf{Case 1:} $e(\cdot) = p(\cdot)$. In this case, $p(X)/e(X) = 1$, and therefore
\begin{align*}
\Theta_1 
&= \E_{\oplus}\left[\gamma_{\mu_1(X),\, m_1(X)}(1)\right] \\
&= \E_{\oplus}[m_1(X)].
\end{align*}
Applying Assumption~3.2(i) with $Y(1)$ yields
\[
\E_{\oplus}[m_1(X)] = \E_{\oplus}[Y(1)],
\]
and hence $\Theta_1 = \E_{\oplus}[Y(1)]$.

\medskip
\noindent\textbf{Case 2:} $\mu_1(\cdot) = m_1(\cdot)$. In this case, for every $X$,
\[
\gamma_{\mu_1(X),\, m_1(X)}\left(\frac{p(X)}{e(X)}\right)
= \gamma_{m_1(X),\, m_1(X)}\left(\frac{p(X)}{e(X)}\right)
= m_1(X),
\]
since a geodesic connecting a point to itself is constant. Therefore,
\[
\Theta_1 = \E_{\oplus}[m_1(X)] = \E_{\oplus}[Y(1)],
\]
where the last equality again follows from Assumption~3.2(i).

\medskip
In both cases, we conclude that $\Theta_1 = \E_{\oplus}[Y(1)]$. The argument for $t=0$ proceeds in the same manner, yielding $\Theta_0 = \E_{\oplus}[Y(0)]$. This completes the proof.
\end{proof}

\section{Proofs for Section 4}

For any positive sequences $a_n, b_n$, we write $a_n \lesssim b_n$ if there is a constant $C >0$ independent of $n$ such that $a_n \leq Cb_n$ for all $n$. 

\begin{proof}[Proof of Theorem 4.1(i)]
Define $\mathrm{diam}(\mathcal{M})= \sup_{\mu_1,\mu_2 \in \mathcal{M}}d(\mu_1,\mu_2)$. By Corollary 3.2.3 in \citet{vaWe23}, it is sufficient to show
\begin{align*}
\sup_{\nu \in \mathcal{M}}\left|Q_{n,t}(\nu; \hat{e}, \hat{\mu}_t) - Q_t(\nu; e, \mu_t)\right| \stackrel{\Prob}{\to} 0,\quad t \in \{0,1\},
\end{align*}
where 
\begin{align*}
Q_t(\nu;e, \mu_t)&=\E\left[d^2\left(\nu,\gamma_{\mu_t(X),Y}\left(\frac{tT}{e(X)} + \frac{(1-t)(1-T)}{1-e(X)}\right)\right)\right],\\
Q_{n,t}(\nu; \hat{e}, \hat{\mu}_t) &= \frac{1}{n}\sum_{i=1}^{n}d^2\left(\nu, \gamma_{\hat{\mu}_t(X_i),Y_i}\left(\frac{tT_i}{\hat{e}(X_i)} + \frac{(1-t)(1-T_i)}{1-\hat{e}(X_i)}\right)\right).
\end{align*}
Now we show $\sup_{\nu \in \mathcal{M}}\left|Q_{n,1}(\nu; \hat{e}, \hat{\mu}_1) - Q_1(\nu; e, \mu_1)\right| \stackrel{\Prob}{\to} 0$;  the proof for $t=0$ is similar. For this, we show that $Q_{n,1}(\nu; \hat{e}, \hat{\mu}_1)$ converges weakly to $Q_1(\nu; e, \mu_1)$ in $\ell^\infty(\mathcal{M})$ and then apply Theorem 1.3.6 in \citet{vaWe23}. From Theorem 1.5.4 in \citet{vaWe23}, this weak convergence follows by showing that 
\begin{itemize}
\item[(i)] $Q_{n,1}(\nu; \hat{e}, \hat{\mu}_1) \stackrel{\Prob}{\to} Q_1(\nu; e, \mu_1)$ for each $\nu \in \mathcal{M}$ as $n \to \infty$.
\item[(ii)] $Q_{n,1}(\nu; \hat{e}, \hat{\mu}_1)$ is asymptotically equicontinuous  in probability, i.e., for each $\varepsilon,\eta>0$, there exists $\delta>0$ such that
\begin{align*}
\limsup_{n \to \infty} \Prob \left(\sup_{d(\nu_1,\nu_2)<\delta}|Q_{n,1}(\nu_1; \hat{e}, \hat{\mu}_1) - Q_{n,1}(\nu_2; \hat{e}, \hat{\mu}_1)|>\varepsilon\right) < \eta.
\end{align*}
\end{itemize}
For an arbitrary  $\nu \in \mathcal{M}$, for (i) we have 
\begin{align*}
&|Q_{n,1}(\nu; \hat{e}, \hat{\mu}_1) - Q_1(\nu; e, \mu_1)|\\
&\leq |Q_{n,1}(\nu; \hat{e}, \hat{\mu}_1) - Q_{n,1}(\nu; e, \hat{\mu}_1)| + |Q_{n,1}(\nu; e, \hat{\mu}_1) - Q_{n,1}(\nu; e, \mu_1)|\\ 
&\quad + |Q_{n,1}(\nu; e, \mu_1) - Q_1(\nu; e, \mu_1)|\\
&=: Q_1 + Q_2 + Q_3
\end{align*}
and observe that 
\begin{align}\label{ineq:Q1}
Q_1 &\leq \frac{1}{n}\sum_{i=1}^{n}\left|d^2\left(\nu, \gamma_{\hat{\mu}_1(X_i),Y_i}\left(\frac{T_i}{\hat{e}(X_i)}\right)\right) - d^2\left(\nu, \gamma_{\hat{\mu}_1(X_i),Y_i}\left(\frac{T_i}{e(X_i)}\right)\right)\right| \nonumber \\
&\leq \frac{2\mathrm{diam}(\mathcal{M})}{n}\sum_{i=1}^{n}d\left( \gamma_{\hat{\mu}_1(X_i),Y_i}\left(\frac{T_i}{\hat{e}(X_i)}\right), \gamma_{\hat{\mu}_1(X_i),Y_i}\left(\frac{T_i}{e(X_i)}\right)\right) \nonumber \\
&\lesssim \frac{2\mathrm{diam}(\mathcal{M})}{n}\sum_{i=1}^{n}\left|\frac{T_i}{\hat{e}(X_i)} - \frac{T_i}{e(X_i)}\right|d(\hat{\mu}_1(X_i),Y_i) \nonumber \\
&\leq \frac{2\mathrm{diam}(\mathcal{M})^2}{\eta_0^2}\sup_{x \in \mathcal{X}}|\hat{e}(x) - e(x)| = O_{\Prob}(\varrho_n), 
\end{align}
\begin{align}\label{ineq:Q2}
Q_2 &\leq \frac{1}{n}\sum_{i=1}^{n}\left|d^2\left(\nu, \gamma_{\hat{\mu}_1(X_i),Y_i}\left(\frac{T_i}{e(X_i)}\right)\right) - d^2\left(\nu, \gamma_{\mu_1(X_i),Y_i}\left(\frac{T_i}{e(X_i)}\right)\right)\right| \nonumber \\
&\leq \frac{2\mathrm{diam}(\mathcal{M})}{n}\sum_{i=1}^{n}d\left( \gamma_{\hat{\mu}_1(X_i),Y_i}\left(\frac{T_i}{e(X_i)}\right), \gamma_{\mu_1(X_i),Y_i}\left(\frac{T_i}{e(X_i)}\right)\right) \nonumber \\
&\leq 2C_0\mathrm{diam}(\mathcal{M})\sup_{x \in \mathcal{X}}d(\hat{\mu}_1(x),\mu_1(x)) = O_{\Prob}(r_n), 
\end{align}
\begin{align}\label{ineq:Q3}
\E[Q_3^2] &\leq \frac{1}{n}\E\left[d^4\left(\nu, \gamma_{\mu_1(X),Y}\left(\frac{T}{e(X)}\right)\right)\right] \leq \frac{\mathrm{diam}^4(\mathcal{M})}{n} = O(n^{-1}). 
\end{align}
Combining (\ref{ineq:Q1}), (\ref{ineq:Q2}), and (\ref{ineq:Q3}), we obtain (i).

Pick any $\nu_1,\nu_2 \in \mathcal{M}$. For (ii), similarly to the argument  leading to  (\ref{ineq:Q1}), 
\begin{align*}
|Q_{n,1}(\nu_1; \hat{e}, \hat{\mu}_1) - Q_{n,1}(\nu_2; \hat{e}, \hat{\mu}_1)| &\leq 2\mathrm{diam}(\mathcal{M})d(\nu_1,\nu_2)
\end{align*}
and this implies $\sup_{d(\nu_1,\nu_2)<\delta}|Q_{n,1}(\nu_1; \hat{e}, \hat{\mu}_1) - Q_{n,1}(\nu_2; \hat{e}, \hat{\mu}_1)| = O_{\Prob}(\delta)$ so that we obtain (ii), which completes the proof.
\end{proof}

\begin{proof}[Proof of Theorem 4.1(ii)]

Note that one can show $\max_{1 \leq k \leq K}d(\hat{\Theta}_{t,k}^{(\mathrm{DR})},\E_\oplus[Y(t)]) = o_{\Prob}(1)$ by applying similar  arguments as in the proof of Theorem 4.1(i). From the definition of $\hat{\Theta}_t^{(\mathrm{CF})}$, 
\begin{align*}
\sum_{k=1}^K\frac{n_k}{n}d^2(\hat{\Theta}_t^{(\mathrm{CF})},\hat{\Theta}_{t,k}^{(\mathrm{DR})}) &\leq \sum_{k=1}^K\frac{n_k}{n}d^2(\E_\oplus[Y(t)],\hat{\Theta}_{t,k}^{(\mathrm{DR})}) \leq \max_{1 \leq k \leq K}d^2(\E_\oplus[Y(t)],\hat{\Theta}_{t,k}^{(\mathrm{DR})}),
\end{align*}
which implies that \[d(\hat{\Theta}_t^{(\mathrm{CF})},\hat{\Theta}_{t,\bar{k}}^{(\mathrm{DR})}) \leq \max_{1 \leq k \leq K}d(\E_\oplus[Y(t)],\hat{\Theta}_{t,k}^{(\mathrm{DR})})\]
for some $\bar{k} \in \{1,\dots,K\}$. Then
\begin{align}
d(\hat{\Theta}_t^{(\mathrm{CF})},\E_\oplus[Y(t)]) &\leq d(\hat{\Theta}_t^{(\mathrm{CF})},\hat{\Theta}_{t,\bar{k}}^{(\mathrm{DR})}) + d(\hat{\Theta}_{t,\bar{k}}^{(\mathrm{DR})},\E_\oplus[Y(t)])\nonumber\\&\leq 2\max_{1 \leq k \leq K}d(\hat{\Theta}_{t,k}^{(\mathrm{DR})}, \E_\oplus[Y(t)])\nonumber\\&= o_{\Prob}(1). \label{ineq:CF-conv}  
\end{align}
\end{proof}

\begin{proof}[Proof of Theorem 4.2(i)]

Now we show (4.6) when $t = 1$. The case  $t=0$ is analogous.  We adapt the proof of Theorem 3.2.5 in \citet{vaWe23}.  Define 
\[
h_{\nu, \tilde{e}, \tilde{\mu}}(x,y,t) = d^2\left(\nu;\gamma_{\tilde{\mu}(x),y}\left(\frac{t}{\tilde{e}(x)}\right)\right) - d^2\left(\Theta_1^{(\mathrm{DR})};\gamma_{\mu_1(x),y}\left(\frac{t}{e(x)}\right)\right).
\]
For some $\delta_0>0$ and $\delta'_1>0$, consider the class of functions 
\[
\mathcal{H}_{\delta_,\delta'_1} := \left\{h_{\nu, \tilde{e}, \tilde{\mu}}(\cdot): d(\nu,\Theta_1^{(\mathrm{DR})})\leq \delta, d_\infty(\tilde{\mu},\mu_1) \leq \delta'_1, \sup_{x\in\mathcal{X}}|\tilde{e}(x)-e(x)|\leq \delta'_1\right\}
\]
for any $\delta\leq\delta_0$. An envelope function for $\mathcal{H}_{\delta,\delta'_1}$ is $2\delta\mathrm{diam}(\mathcal{M})$. 
Observe that 
\begin{align*}
&|h_{\nu', \tilde{e}', \tilde{\mu}'}(x,y,t) - h_{\nu, \tilde{e}, \tilde{\mu}}(x,y,t)|\\
&\leq 2\mathrm{diam}(\mathcal{M})\left(1 + C_0 + \frac{C_e\mathrm{diam}(\mathcal{M})}{\eta_0^2}\right)\left\{d(\nu',\nu) + |\tilde{e}' - e| +  d_\infty(\tilde{\mu}',\tilde{\mu})\right\}. 
\end{align*}
Then for $\varepsilon>0$, following Theorem 2.7.17 of \citet{vaWe23}, the $\delta \varepsilon$ bracketing number of $\mathcal{H}_{\delta,\delta'_1}$ is bounded by a multiple of 
\[N(c_1\delta\varepsilon, B_\delta(\Theta_1^{(\mathrm{DR})}),d)N(c_2\delta\varepsilon, B_{\delta'_1}(\mu_1),d_\infty)(\delta \varepsilon)^{-c_3},\]
where $c_1,c_2$, and $c_3$ are constants depending only on $p$, $\mathrm{diam}(\mathcal{M})$, $C_e$, and  $C_0$. 

Assumption 4.4(i) and Theorem 2.14.16 of \citet{vaWe23} imply that, for small enough $\delta>0$,
\begin{align}
\E\Bigg[\sup_{\substack{d(\nu,\Theta_1)\leq \delta, \sup_{x\in\mathcal{X}}|\tilde{e}(x)-e(x)|\leq \delta'_1 \\ 
d_\infty(\tilde{\mu},\mu_1)\leq \delta'_1}} 
\Big| & Q_{n,1}(\nu; \tilde{e}, \tilde{\mu}) - Q_{n,1}(\Theta_1^{(\mathrm{DR})}; \tilde{e}, \tilde{\mu}) \nonumber \\
& - Q_1(\nu; \tilde{e}, \tilde{\mu}) + Q_1(\Theta_1^{(\mathrm{DR})}; \tilde{e}, \tilde{\mu}) \Big| \Bigg]\nonumber\\
&\lesssim \frac{2\delta \mathrm{diam}(\mathcal{M})}{\sqrt{n}}\left(J_1(\delta) + J_{\mu_1}(\delta) + \int_0^1 \sqrt{-\log (\delta \varepsilon)}d\varepsilon \right) \nonumber\\&\lesssim \frac{\delta (1 + \delta^{-\vartheta})}{\sqrt{n}} \lesssim \frac{\delta^{1-\vartheta}}{\sqrt{n}}. \label{ineq:DV-Ass(B2)}
\end{align}

For any $\epsilon \in (0,1)$, set $q_n^{-1} = \max\{n^{-\frac{\beta}{4(\beta-1+\vartheta)}}, (\varrho_n + r_n)^{\frac{\beta 1-\epsilon}{2(\beta-1)}} \}$ and 
\[
S_{j,n} = \left\{\nu: 2^{j-1} < q_nd(\nu,\Theta_1^{(\mathrm{DR})})^{\beta/2} \leq 2^j\right\}.
\]
Choose $\eta>0$ to satisfy Assumption 4.4(ii) and also small enough so that Assumption 4.4(i) holds for all $\delta <\eta$ and set $\tilde{\eta} = \eta^{\beta/2}$. For any $L, \delta_1, K>0$, 
\begin{align}\label{ineq:GDR-rate}
&\Prob\left(q_n d(\hat{\Theta}_1^{(\mathrm{DR})},\Theta_1^{(\mathrm{DR})})^{\beta/2} > 2^L\right) \nonumber \\ 
&\leq \Prob(A_n^c) + \Prob(2d(\hat{\Theta}_1^{(\mathrm{DR})},\Theta_1^{(\mathrm{DR})}) \geq \eta) \nonumber \\
&\quad + \sum_{\substack{j \geq L \\ 2^j \leq q_n \tilde{\eta}}}\!\!\!\Prob \!\left(\!\left\{\sup_{\nu \in S_{j,n}}\!\!\!\left\{Q_{n,1}(\nu;\hat{e}, \hat{\mu}_1) - Q_{n,1}(\Theta_1^{(\mathrm{DR})};\hat{e}, \hat{\mu}_1)\right\}\! \leq Kq_n^{-2}\right\}\! \cap A_n \!\right),
\end{align}
where $A_n = \{d_\infty(\hat{\mu}_1,\mu_1) \leq \delta_1 q_n^{-\frac{2(\beta-1)}{\beta}}, \|\hat{\varphi} - \varphi_*\| \leq \delta_1 q_n^{-\frac{2(\beta-1)}{\beta}}\}$.

Note that for any $\varepsilon>0$, there exist an integer $n_\varepsilon$ such that 
\[\Prob(A_n^c) < \varepsilon/3,\;\Prob(2d(\hat{\Theta}_1^{(\mathrm{DR})},\Theta_1^{(\mathrm{DR})})\geq \eta) < \varepsilon/3,\text{ and }\delta_1q_n^{-\frac{2(\beta-1)}{\beta}}\leq \delta'_1.\] 
Define $A_{n,\varepsilon}$ as $A_n$ with $n \geq n_\varepsilon$. Then, on $A_{n,\varepsilon}$, for each fixed $j$ such that $2^j \leq q_n \tilde{\eta}$, one has for all $\nu \in S_{j,n}$, 
\begin{align*}
&Q_{n,1}(\nu;\hat{e}, \hat{\mu}_1) - Q_{n,1}(\Theta_1^{(\mathrm{DR})};\hat{e}, \hat{\mu}_1) \\
&\geq Q_1(\nu;\hat{e}, \hat{\mu}_1) - Q_1(\Theta_1^{(\mathrm{DR})};\hat{e}, \hat{\mu}_1)-\sup_{\substack{d(\nu,\Theta_1^{(\mathrm{DR})})^{\beta/2} \leq \frac{2^j}{q_n}}} 
\Big| Q_{n,1}(\nu;\hat{e}, \hat{\mu}_1) - Q_{n,1}(\Theta_1^{(\mathrm{DR})};\hat{e}, \hat{\mu}_1) \\
&\quad - Q_1(\nu;\hat{e}, \hat{\mu}_1) + Q_1(\Theta_1^{(\mathrm{DR})};\hat{e}, \hat{\mu}_1) \Big|\\
&\geq \frac{2^{2j-2}}{q_n^2}\left(\frac{C}{2} - C_1\delta_1^{\frac{\beta}{2(\beta-1)}}2^{2-j}\right)-\sup_{\substack{d(\nu,\Theta_1^{(\mathrm{DR})})^{\beta/2} \leq \frac{2^j}{q_n}}} 
\Big|  Q_{n,1}(\nu;\hat{e}, \hat{\mu}_1) - Q_{n,1}(\Theta_1^{(\mathrm{DR})};\hat{e}, \hat{\mu}_1) \\
&\quad - Q_1(\nu;\hat{e}, \hat{\mu}_1) + Q_1(\Theta_1^{(\mathrm{DR})};\hat{e}, \hat{\mu}_1) \Big|.
\end{align*}
Then for large $j$ such that $C/2 - C_1\delta_1^{\frac{\beta}{2(\beta-1)}} 2^{2-j}
 - K 2^{2-2j} \geq C/4$, we have
\begin{align}\label{ineq:GDR-rate-2}
&\Prob\left(\left\{\sup_{\nu \in S_{j,n}}\left\{Q_{n,1}(\nu;\hat{e}, \hat{\mu}_1) - Q_{n,1}(\Theta_1^{(\mathrm{DR})};\hat{e}, \hat{\mu}_1)\right\} \leq Kq_n^{-2}\right\} \cap A_{n,\varepsilon}\right) \nonumber \\
&\leq \Prob \Bigg( \sup_{\substack{d(\nu,\Theta_1^{(\mathrm{DR})})^{\beta/2} \leq \frac{2^j}{q_n}, \\
d_\infty(\tilde{\mu},\mu_1) \leq \delta_1 q_n^{-2(\beta-1)/\beta}, \\
\sup_{x\in\mathcal{X}}|\tilde{e}(x)-e(x)| \leq \delta_1 q_n^{-2(\beta-1)/\beta}}} \Bigg| 
Q_{n,1}(\nu; \tilde{e}, \tilde{\mu}) - Q_{n,1}(\Theta_1^{(\mathrm{DR})}; \tilde{e}, \tilde{\mu})\nonumber \\
& \quad\quad\quad - Q_1(\nu; \tilde{e}, \tilde{\mu}) + Q_1(\Theta_1^{(\mathrm{DR})}; \tilde{e}, \tilde{\mu}) 
\Bigg| \geq \frac{C 2^{2j-2}}{4q_n^2} \Bigg).
\end{align}

For each $j$, in the sum on the right-hand side of (\ref{ineq:GDR-rate}), we have $d(\nu, \Theta_1^{(\mathrm{DR})}) \leq (2^j/q_n)^{2/\beta} \leq \eta$.  Using (\ref{ineq:GDR-rate-2}) and the Markov inequality,
\begin{align*}
&\Prob\left(q_n d(\hat{\Theta}_1^{(\mathrm{DR})},\Theta_1^{(\mathrm{DR})})^{\beta/2} > 2^L\right)\\
&\leq\frac{4}{C} \sum_{\substack{j \geq L \\ 2^j \leq q_n \tilde{\eta}}} 
\frac{2^{-2j}}{q_n^{-2}} \E \Bigg[ 1_{A_{n,\varepsilon}} 
\sup_{\nu \in S_{j,n}} \Big| Q_{n,1}(\nu; \tilde{e}, \tilde{\mu}) - Q_{n,1}(\Theta_1^{(\mathrm{DR})}; \tilde{e}, \tilde{\mu}) \\
&\quad - Q_1(\nu; \tilde{e}, \tilde{\mu}) + Q_1(\Theta_1^{(\mathrm{DR})}; \tilde{e}, \tilde{\mu}) \Big| \Bigg] 
+ \frac{2\varepsilon}{3}\\
&\lesssim \sum_{\substack{j \geq L \\ 2^j \leq q_n \tilde{\eta}}}\frac{2^{-2j + \frac{2j}{\beta}(1-\vartheta)}}{q_n^{-2 + \frac{2}{\beta}(1-\vartheta)}\sqrt{n}} + \frac{2\varepsilon}{3} \\&\leq \sum_{j \geq L}\left(\frac{1}{4^{\frac{\beta - 1 + \vartheta}{\beta}}}\right)^j + \frac{2\varepsilon}{3}.
\end{align*}
Since $\beta>1$, the last series converges and hence the first term on the far right-hand side can be made smaller than $\varepsilon>0$ for large $L$. Therefore, we obtain the desired result, $$d(\hat{\Theta}_1^{(\mathrm{DR})},\E_\oplus[Y(1)]) = O_{\Prob}(q_n^{-\frac{2}{\beta}}) 
 =O_{\Prob}(n^{-\frac{1}{2(\beta - 1 +\vartheta)}} + (\varrho_n + r_n)^{\frac{1-\epsilon}{(\beta-1)}}).$$
\end{proof}

\begin{proof}[Proof of Theorem 4.2(ii)]
Applying almost the same argument as in the proof of Theorem 4.2(i), one can show 
\[
\max_{1 \leq k \leq K}d(\hat{\Theta}_{t,k}^{(\mathrm{DR})},\E_\oplus[Y(t)]) = O_{\Prob}\left(n^{-\frac{1}{2(\beta - 1 +\vartheta)}} + (\varrho_n + r_n)^{\frac{1-\epsilon}{(\beta-1)}}\right).
\]
Then (\ref{ineq:CF-conv}) yields the result.
\end{proof}

\section{Proofs for Section 5}
\begin{proof}[Proof of Proposition 5.1]
We only prove the result for $\hat{\Theta}_{t'}^{(\mathrm{DR})}$ since the proof for $\hat{\Theta}_{t'}^{(\mathrm{CF})}$ is almost the same. Observe that
\begin{align*}
&T_{m,n}^{(t)}(\hat{\Theta}_{t'}^{(\mathrm{DR})}) - \tilde{T}_m^{(t)}(\Eo[Y(t')])\\
&= T_{m,n}^{(t)}(\hat{\Theta}_{t'}^{(\mathrm{DR})})- T_{m,n}^{(t)}(\Eo[Y(t')]) + T_{m,n}^{(t)}(\Eo[Y(t')]) - \tilde{T}_m^{(t)}(\Eo[Y(t')])\\
&= \frac{1}{\sqrt{m}}\sum_{i=1}^m\left\{d^2\left(\hat{\Theta}_{t'}^{(\mathrm{DR})}, \gamma_{\hat{\mu}_t(X_i), Y_i}\left(\frac{tT_i}{\hat{e}(X_i)} + \frac{(1-t)(1 - T_i)}{1 - \hat{e}(X_i)}\right)\right) \right.\\
&\qquad \qquad \left.-d^2\left(\Eo[Y(t')], \gamma_{\mu_t(X_i), Y_i}\left(\frac{tT_i}{e(X_i)} + \frac{(1-t)(1 - T_i)}{1 - e(X_i)}\right)\right)\right\}\\
&\quad +\frac{1}{\sqrt{m}}\sum_{i=1}^m\left\{d^2\left(\Eo[Y(t')], \gamma_{\hat{\mu}_t(X_i), Y_i}\left(\frac{tT_i}{\hat{e}(X_i)} + \frac{(1-t)(1 - T_i)}{1 - \hat{e}(X_i)}\right)\right) \right.\\
&\qquad \qquad \left.-d^2\left(\Eo[Y(t')], \gamma_{\mu_t(X_i), Y_i}\left(\frac{tT_i}{e(X_i)} + \frac{(1-t)(1 - T_i)}{1 - e(X_i)}\right)\right)\right\}\\
&\quad + \sqrt{m}\left\{Q_{n,t}(\hat{\Theta}_t^{(\mathrm{DR})};\hat{e}, \hat{\mu}_t) - Q_t(\Eo[Y(t)];e, \mu_t)\right\}\\
&=: I_1^{(t)} + I_2^{(t)} + \sqrt{m}I_3^{(t)}.
\end{align*}
Applying a similar argument in the proof of Theorem 4.1 (i), we have
\begin{align*}
|I_1^{(t)}| &=  O_{\Prob}(\sqrt{m}\{n^{-\frac{1}{2(\beta - 1 +\vartheta)}} + (\varrho_n + r_n)^{\frac{1-\epsilon}{\beta-1}}+\varrho_n + r_n\}) = o_{\Prob}(1),\\
|I_2^{(t)}| &=O_{\Prob}(\sqrt{m}\{\varrho_n + r_n\})=o_{\Prob}(1),\\
\sqrt{m}|I_3^{(t)}| &=  O_{\Prob}(\sqrt{m}\{n^{-\frac{1}{2(\beta - 1 +\vartheta)}} + (\varrho_n + r_n)^{\frac{1-\epsilon}{\beta-1}} + \varrho_n + r_n + n^{-1/2}\}) = o_{\Prob}(1).
\end{align*}
Therefore, we have $T_{m,n}^{(t)}(\hat{\Theta}_{t'}^{(\mathrm{DR})}) - \tilde{T}_m^{(t)}(\Eo[Y(t')])=  o_{\Prob}(1)$ as $n \to \infty$. The conclusion follows immediately from the asymptotic properties of $\tilde{T}_m^{(t)}(\Eo[Y(t')])$. 
\end{proof}

\begin{proof}[Proof of Proposition 5.2]
As in Proposition 5.1, we present the result only for $\hat{\Theta}_t^{(\mathrm{DR})}$. Observe that under the assumptions of Proposition 5.1, one can see
\begin{align*}
&T_{m,n}^{(t,*)}(\hat{\Theta}_t^{(\mathrm{DR})})\\ 
&= \frac{1}{\sqrt{m}}\sum_{i=1}^m \xi_i\left\{d^2\left(\Eo[Y(t)], \gamma_{\mu_t(X_i), Y_i}\left(\frac{tT_i}{e(X_i)} + \frac{(1-t)(1 - T_i)}{1 - e(X_i)}\right)\right)\right.\\
&\qquad \qquad \qquad \left.-Q_{m,t}(\Eo[Y(t)];e,\mu_t)\right\} + o_{\Prob}(1).
\end{align*}
This and the conditional multiplier central limit theorem (e.g., Theorem 2.1, \citet{pauly:11}) yield $T_{m,n}^{(t,*)}(\hat{\Theta}_t^{(\mathrm{DR})}) \stackrel{\Prob_\xi}{\leadsto} N(0,V^{(t)}(\Eo[Y(t)]))$ where $\stackrel{\Prob_\xi}{\leadsto}$ denote conditional weak convergence given the data; see Section~1.13 of \citet{vaWe23}. By applying the conditional continuous mapping theorem (e.g., Theorem 10.8, \citet{koso:08}), we have $|T_{m,n}^{(t,*)}(\hat{\Theta}_t^{(\mathrm{DR})})| \stackrel{\Prob_\xi}{\leadsto} |N(0,V^{(t)}(\Eo[Y(t)]))|$ and this yields the desired result.
\end{proof}

\section{Brain functional connectivity networks and Alzheimer's disease}\label{subsec:Alzheimer}
Resting-state functional Magnetic Resonance Imaging (fMRI) methodology allows for the study of brain activation and the identification of brain regions exhibiting similar activity during the resting state \citep{ferr:13, sala:15}. In resting-state fMRI, a time series of Blood Oxygen Level Dependent (BOLD) signals are obtained for different regions of interest (ROI). The temporal coherence between pairwise ROIs is typically measured by temporal Pearson correlation coefficients (PCC) of the fMRI time series, forming an $m\times m$ correlation matrix when considering $m$ distinct ROIs. Alzheimer's disease has been found to be associated with anomalies in the functional integration and segregation of ROIs \citep{zhan:10, damo:12}.

Our study utilized data obtained from the Alzheimer's Disease Neuroimaging Initiative (ADNI) database (\url{adni.loni.usc.edu}). The dataset in our analysis consists of 372 cognitively normal (CN), and 145 Alzheimer's (AD) subjects. We used the automated anatomical labeling (AAL) atlas \citep{tzou:02} to parcellate the whole brain into 90 ROIs, with 45 ROIs in each hemisphere, so that $m=90$. Details about the ROIs can be found in Table~\ref{tab:aal}. The preprocessing of the BOLD signals was implemented by adopting standard procedures of slice-timing correction, head motion correction, and other standard steps. A PCC matrix was calculated for all time series pairs for each subject. These matrices were then converted into simple, undirected, weighted networks by setting diagonal entries to 0 and thresholding the absolute values of the remaining correlations. Density-based thresholding was employed, where the threshold varied from subject to subject to achieve a desired, fixed connection density. Specifically, we retained the 15\% strongest connections. 

In our analysis, Alzheimer's disease, coded as a binary variable (0 for cognitively normal, 1 for diagnosed with Alzheimer's disease), serves as exposure of interest. The outcomes are brain functional connectivity networks, where two confounders are the age and gender of each subject. To quantify the impact of Alzheimer's disease on brain functional connectivity, we computed the entry-wise differences between the adjacency matrices of the two mean potential networks for the DR estimator, as depicted in Figure~\ref{fig:adniate}. Our observations reveal that Alzheimer's disease leads to notable effects on various regions of the human brain. Specifically, the central region, orbital surface in the frontal lobe, temporal lobe, occipital lobe, and subcortical areas exhibit more pronounced influences, displaying clustered patterns consistent with findings in previous literature \citep{plan:22}. Notably, damage to the frontal lobe is associated with issues in judgment, intelligence, and behavior, while damage to the temporal lobe can impact memory.

\begin{figure}[tb]
    \centering
    \includegraphics[width=0.8\linewidth]{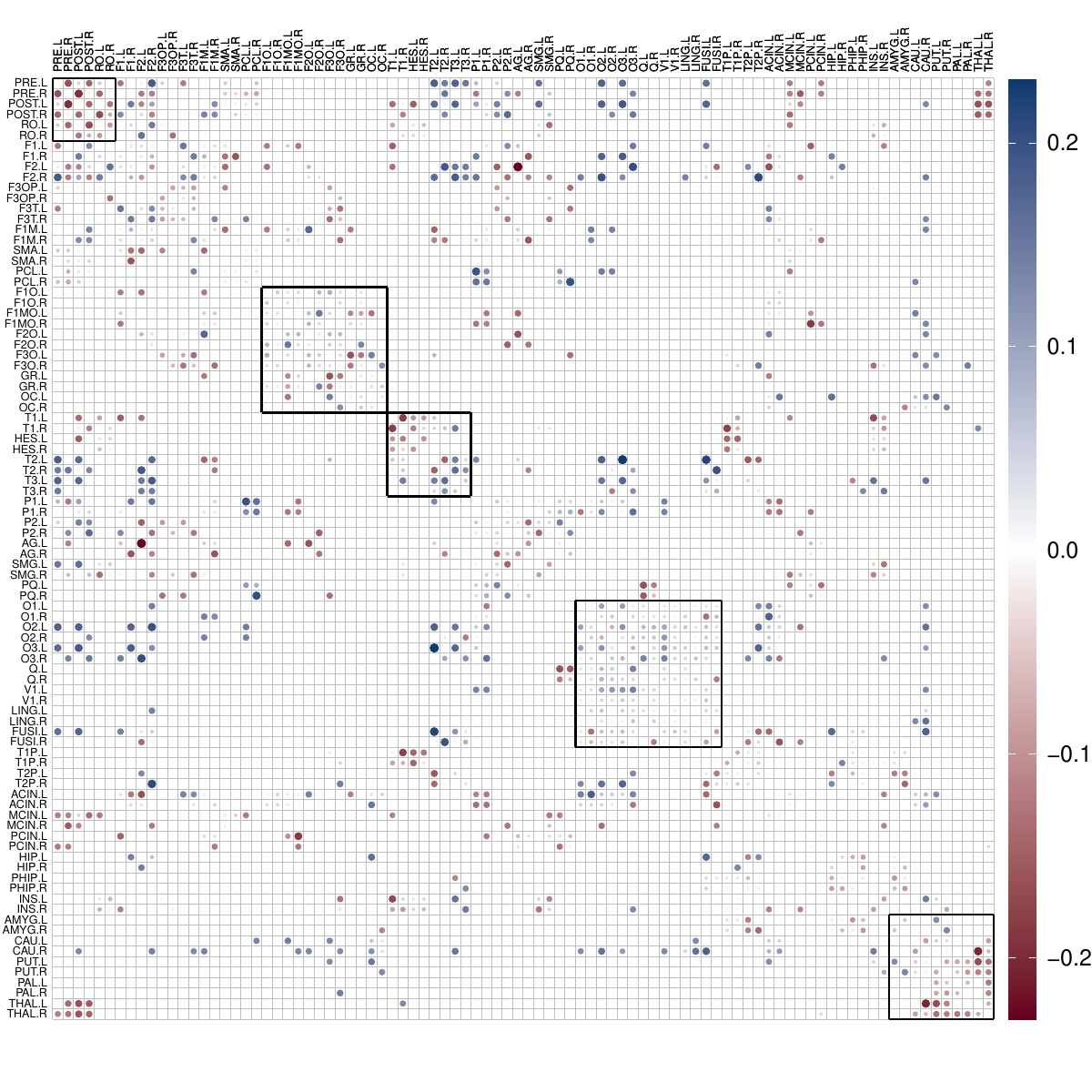}
    \caption{Entry-wise differences between the adjacency matrices corresponding to the mean potential networks for cognitively normal and Alzheimer's disease subjects, estimated using the doubly robust procedure. The five regions circled by squares, from top left to bottom right, correspond to the central region, orbital surface in the frontal lobe, temporal lobe, occipital lobe, and subcortical areas, respectively.}
\label{fig:adniate}
\end{figure}

To visualize the GATE in this setting, we focus on a $6\times 6$ submatrix corresponding to central brain regions, which exhibit some of the strongest disease-related changes. Figure~\ref{fig:adnigds} displays the intrinsic geodesic connecting the DR-estimated mean potential networks for cognitively normal and Alzheimer's disease subjects within this subnetwork. Each panel shows an intermediate network along the geodesic at $t \in \{0, 0.2, 0.4, 0.6, 0.8, 1\}$, where $t=0$ corresponds to the mean potential network for cognitively normal subjects and $t=1$ corresponds to that for Alzheimer's disease. The geodesic provides an interpretable depiction of the causal effect, revealing a smooth and coordinated weakening of connections in these central regions as one moves from $t=0$ to $t=1$. This pattern reflects a gradual loss of functional integration, consistent with disruptions of executive and memory-related pathways documented in the Alzheimer's literature. Standard approaches based on elementwise differences or summary network measures capture only isolated aspects of this effect, whereas the geodesic highlights how multiple connections degrade jointly and provides a coherent representation of how the overall network structure changes under disease progression.

\begin{figure}[tb]
    \centering
	\includegraphics[width=0.8\linewidth]{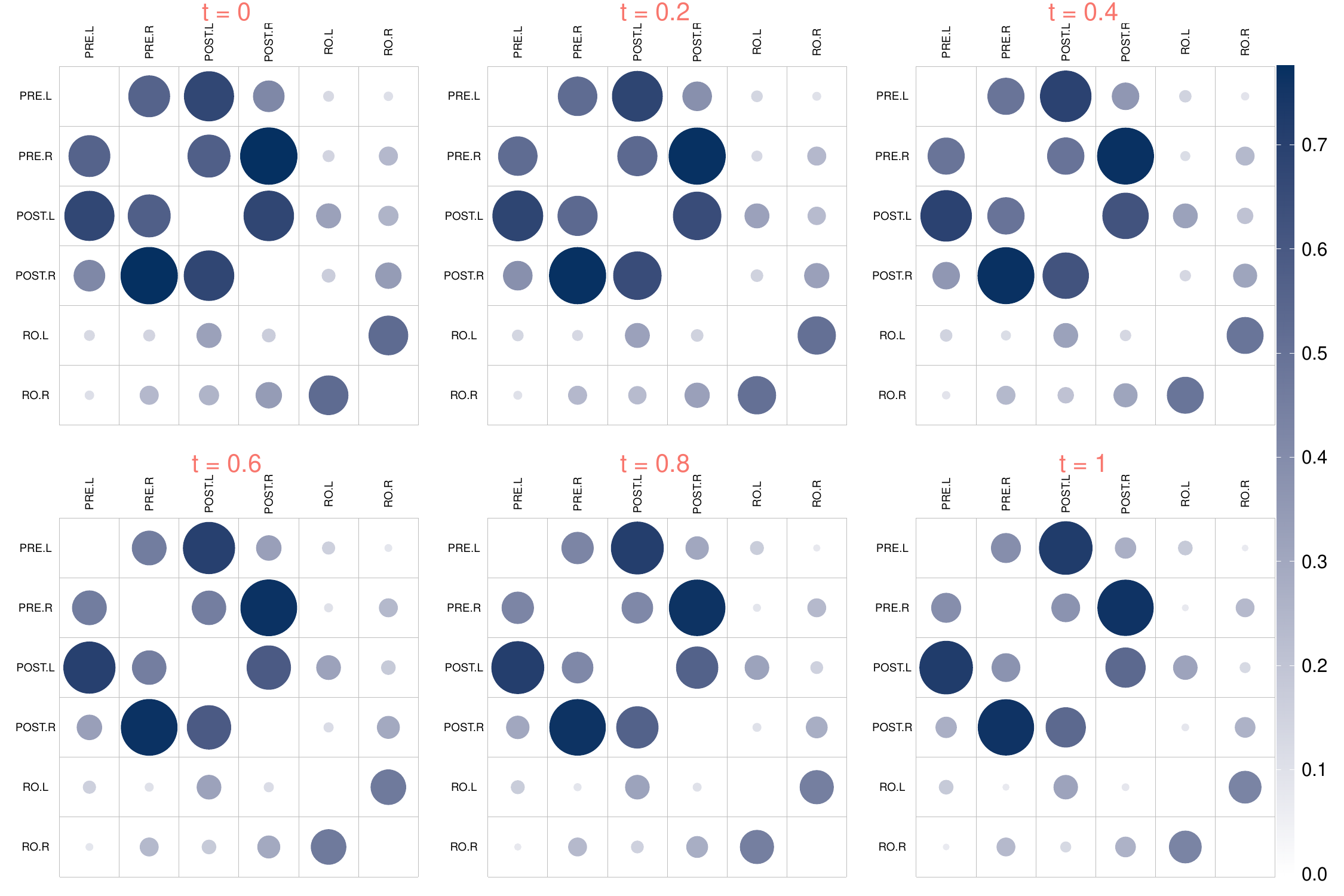}
    \caption{GATE for the ADNI functional connectivity example. Shown are six intermediate networks along the intrinsic geodesic between the DR-estimated mean potential networks for cognitively normal ($t=0$) and Alzheimer's disease ($t=1$) subjects, restricted to a $6\times6$ submatrix of central brain regions. The geodesic displays the coordinated weakening of connectivity associated with Alzheimer's disease.}
    \label{fig:adnigds}
\end{figure}

To assess structural changes in brain functional connectivity networks following Alzheimer's disease, Table~\ref{tab:adninm} provides a summary of commonly adopted network measures that quantify functional integration and segregation, with definitions provided in \citet{rubi:10}. Our observations indicate a decrease in all functional integration and segregation measures of brain functional connectivity networks, with local efficiency experiencing the most substantial reduction. Functional segregation, denoting the brain's ability for specialized processing within densely interconnected groups of regions, and functional integration, representing the capacity to rapidly combine specialized information from distributed brain regions, are both compromised. This finding implies that Alzheimer's disease leads to lower global and local efficiency in human brains, disrupting crucial functions such as judgment and memory \citep{ahma:23}.

\begin{table}[tb]
	\caption{Network measures of mean potential networks for the doubly robust estimator}
	\centering
	\begin{tabular}{llcc}
		\hline
            & &Cognitively normal & Alzheimer's disease\\\hline
		\multirow{2}{*}{Functional integration}&Global efficiency & 3.117 & 2.914\\
            &Characteristic path length & 10.319 & 9.593\\
            \multirow{2}{*}{Functional segregation}&Local efficiency & 3.599 & 3.035\\
            &Clustering coefficient & 0.154 & 0.142\\
            \hline
	\end{tabular}
	\label{tab:adninm}
\end{table}
\newpage
{\small
\begin{longtable}{|lll|}
\hline
ROI & Lobe & Label \\ \hline
Central region &  &  \\ 
\hspace{0.5cm}1 & Precentral gyrus & PRE \\ 
\hspace{0.5cm}2 & Postcentral gyrus & POST \\ 
\hspace{0.5cm}3 & Rolandric operculum & RO \\ 

Frontal lobe &  &  \\ 
\hspace{0.5cm}Lateral surface &  &  \\ 
\hspace{0.5cm}4 & Superior frontal gyrus, dorsolateral & F1 \\ 
\hspace{0.5cm}5 & Middle frontal gyrus & F2 \\ 
\hspace{0.5cm}6 & Inferior frontal gyrus, opercular part & F3OP \\ 
\hspace{0.5cm}7 & Inferior frontal gyrus, triangular part & F3T \\ 

\hspace{0.5cm}Medial surface &  &  \\ 
\hspace{0.5cm}8 & Superior frontal gyrus, medial & F1M \\ 
\hspace{0.5cm}9 & Supplementary motor area & SMA \\ 
\hspace{0.5cm}10 & Paracentral lobule & PCL \\ 

\hspace{0.5cm}Orbital surface &  &  \\ 
\hspace{0.5cm}11 & Superior frontal gyrus, orbital part & F1O \\ 
\hspace{0.5cm}12 & Superior frontal gyrus, medial orbital & F1MO \\ 
\hspace{0.5cm}13 & Middle frontal gyrus, orbital part & F2O \\ 
\hspace{0.5cm}14 & Inferior frontal gyrus, orbital part & F3O \\ 
\hspace{0.5cm}15 & Gyrus rectus & GR \\ 
\hspace{0.5cm}16 & Olfactory cortex & OC \\ 

Temporal lobe &  &  \\ 
\hspace{0.5cm}Lateral surface &  &  \\ 
\hspace{0.5cm}17 & Superior temporal gyrus & T1 \\ 
\hspace{0.5cm}18 & Heschl gyrus & HES \\ 
\hspace{0.5cm}19 & Middle temporal gyrus & T2 \\ 
\hspace{0.5cm}20 & Inferior temporal gyrus & T3 \\ 

Parietal lobe &  &  \\ 
\hspace{0.5cm}Lateral surface &  &  \\ 
\hspace{0.5cm}21 & Superior parietal gyrus & P1 \\ 
\hspace{0.5cm}22 & Inferior parietal & P2 \\ 
\hspace{0.5cm}23 & Angular gyrus & AG \\ 
\hspace{0.5cm}24 & Supramarginal gyrus & SMG \\ 

\hspace{0.5cm}Medial surface &  &  \\ 
\hspace{0.5cm}25 & Precuneus & PQ \\ 

Occipital lobe &  &  \\ 
\hspace{0.5cm}Lateral surface &  &  \\ 
\hspace{0.5cm}26 & Superior occipital gyrus & O1 \\ 
\hspace{0.5cm}27 & Middle occipital gyrus & O2 \\ 
\hspace{0.5cm}28 & Inferior occipital gyrus & O3 \\ 

\hspace{0.5cm}Medial surface &  &  \\ 
\hspace{0.5cm}29 & Cuneus & Q \\ 
\hspace{0.5cm}30 & Calcarine fissure & V1 \\ 
\hspace{0.5cm}31 & Lingual gyrus & LING \\ 
\hspace{0.5cm}32 & Fusiform gyrus & FUSI \\ 

Limbic lobe &  &  \\ 
\hspace{0.5cm}33 & Temporal pole: superior temporal gyrus & T1P \\ 
\hspace{0.5cm}34 & Temporal pole: middle temporal gyrus & T2P \\ 
\hspace{0.5cm}35 & Anterior cingulate and paracingulate gyri & ACIN \\ 
\hspace{0.5cm}36 & Median cingulate and paracingulate gyri & MCIN \\ 
\hspace{0.5cm}37 & Posterior cingulate gyrus & PCIN \\ 
\hspace{0.5cm}38 & Hippocampus & HIP \\ 
\hspace{0.5cm}39 & Parahippocampal gyrus & PHIP \\ 
\hspace{0.5cm}40 & Insula & INS \\ 

Subcortical &  &  \\ 
\hspace{0.5cm}41 & Amygdala & AMYG \\ 
\hspace{0.5cm}42 & Caudate nuclei & CAU \\
\hspace{0.5cm}43 & Lenticular nucleus, putamen & PUT \\
\hspace{0.5cm}44 & Lenticular nucleus, pallidum & PAL \\
\hspace{0.5cm}45 & Thalamus & THAL \\
\hline
\caption{Anatomical regions of interest (ROIs) in each hemisphere for the AAL atlas \citep{tzou:02}.}
\label{tab:aal}
\end{longtable}
}

\bibliography{collection}
\bibliographystyle{apalike}

\end{document}